%% file: aliasJASA.tex
\documentclass[12pt]{article}
\usepackage{amsmath}
\usepackage{graphicx}
\usepackage{enumerate}
\usepackage{natbib}
\usepackage{url} 
\usepackage{xcolor}
\usepackage{hyperref}

\usepackage{setspace}

\newcommand{\blind}{1}

\newtheorem{definition}{Definition}
\newtheorem{proposition}{Proposition}
\newenvironment{proof}[1][Proof]{\noindent\textbf{#1.} }{\ \rule{0.5em}{0.5em}}

\newcommand\indep{\perp\!\!\!\perp}
\newcommand*{\inlineequation}[2][]{%
  \begingroup
    \refstepcounter{equation}%
    \ifx\\#1\\%
    \else
      \label{#1}%
    \fi
    \relpenalty=10000 %
    \binoppenalty=10000 %
    \ensuremath{%
      #2%
    }%
    ~\@eqnnum
  \endgroup
}

\begin{document}

\def\spacingset#1{\renewcommand{\baselinestretch}%
{#1}\small\normalsize} \spacingset{1}


\date{} 

\if1\blind
{
  \title{\bf Effect Aliasing in Observational Studies}
  \author{Paul R. Rosenbaum \\
    Department of Statistics and Data Science\\Wharton School, University of Pennsylvania\\
    and \\
    Jos\'{e} R. Zubizarreta\thanks{
    The authors gratefully acknowledge funding from the Patient Centered Outcomes Research Initiative (PCORI, award ME-2022C1-
25648).}\hspace{.2cm}\\
    Departments of Health Care Policy, Biostatistics, and Statistics \\ Harvard University}
  \maketitle
} \fi

\if0\blind
{
  \bigskip
  \bigskip
  \bigskip
  \begin{center}
    {\LARGE\bf Effect Aliasing in Observational Studies}
\end{center}
  \medskip
} \fi

\bigskip
\begin{abstract}
In experimental design, aliasing of effects occurs in fractional factorial
experiments, where certain low order factorial effects are indistinguishable
from certain high order interactions: low order contrasts may be orthogonal
to one another, while their higher order interactions are aliased and not identified. \ In
observational studies, aliasing occurs when certain combinations of
covariates --- e.g., time period and various eligibility criteria for
treatment --- perfectly predict the treatment that an individual will
receive, so a covariate combination is aliased with a particular treatment.
\ In this situation, when a contrast among several groups is used to estimate a
treatment effect, collections of individuals defined by contrast weights may be balanced with respect to summaries of low-order interactions between covariates and treatments, but necessarily
not balanced with respect to summaries of high-order interactions between covariates and treatments. \ We develop a
 theory of aliasing in observational studies, illustrate that theory in
an observational study whose aliasing is more robust than conventional
difference-in-differences, and develop a new form of matching 
to construct balanced confounded factorial designs from observational data.
\end{abstract}

\noindent%
\vfill

\newpage
\spacingset{1} 

\section{Introduction: Aliasing in Experiments and Observational Studies}
\spacingset{1.84}
\vspace{-.5cm}
\subsection{Experiments and Observational Studies}

In his discussion of the planning of observational studies, 
\citet[p.~236]{cochran1965planning}
wrote:

\begin{quotation}
\noindent \ldots to a large extent, workers in observational research have
tried to copy devices that have proved effective in controlled experiments
\ldots\ Dorn (1953) recommended that the planner of an observational study
always ask himself the question, `How would the study be conducted if it
were possible to do it by controlled experimentation?'
\end{quotation}

\noindent For instance, blocking and adjustments appear in both experiments
and observational studies. \ 
\citet[p.~237]{cochran1965planning}%
\ continued:

\begin{quotation}
\noindent There is, however, an important difference between the demands
made on blocking or adjustments in controlled experiments and in
observational studies. In controlled experiments, the skilful use of
randomization protects against most types of bias arising from disturbing
variables. \ Consequently, the function of blocking or adjustment is to
increase precision rather than to guard against bias. In observational
studies, in which no random assignment of subjects to comparison groups is
possible, blocking and adjustment take on the additional role of protecting
against bias. Indeed, this is often their primary role.
\end{quotation}

Here, in an observational study, we discuss, implement and illustrate two
other important tools from experimental design, namely deliberate aliasing
of effects and deliberate balanced confounding of effects with blocks to
produce smaller, more homogeneous blocks. \ We extract from a large
administrative record system a small observational study of 16,800 people in
which $2^{3}$ factor combinations are studied in 4200 balanced blocks of
size 4, with some unavoidable aliasing of effects, while balancing many
observed covariates and covariate interactions, including potential effect modifiers.

\subsection{Aliasing in the Simplest Fractional Factorial Experiment}

Aliasing is a concept from experimental design that also appears, for
different reasons, in some observational studies. \ A large factorial
experiment may be run in fractional replication: perhaps $K$ factors, each
at 2 levels, are studied in fewer than $2^{K}$ observations, by carefully
selecting $2^{K-p}$ of the $2^{K}$ possible factor combinations for use in
the experiment. \ The price of fractional replication is that certain
effects become indistinguishable or aliased. \ The art of designing a
fractional factorial experiment is to alias effects thought to be important
with effects thought to be negligible, often aliasing main effects of
factors with high-order interactions among factors, or aliasing certain
interactions with other interactions; see, for instance \cite%
{wu2011experiments}.

The simplest and smallest nontrivial example of a fractional factorial
experiment is a $\frac{1}{2}$-fraction of a $2\times 2\times 2=2^{3}$
experiment in factors $A$, $B$, $C$, observations $i=1,\ldots ,\,8$, in
which observations 1, 4, 6 and 7 are the $\frac{1}{2}$-fraction, as
displayed in Table 1.

The complete $2^{3}$ factorial can be run in the $2^{3}=8$ observations on
the left in Table 1, and the $\frac{1}{2}$-fractional factorial can be run
in the four observations on the right in Table 1. \ In the full $2^{3}$
factorial, there is one linear contrast in these eight observations that,
under Gauss-Markov assumptions for a linear model, is the best linear
unbiased estimate of each of the 7 factorial effects, namely the 3 main
effects, the 3 two-factor interactions, and the one three factor
interaction. \ For instance, the main effect of $A$ is estimated by $\frac{1%
}{4}\left( y_{1}+y_{2}+y_{3}+y_{4}-y_{5}-y_{6}-y_{7}-y_{8}\right) $ with
contrast weights $\frac{1}{4}\left( 1,1,1,1,-1,-1,-1,-1\right) $, while the $%
ABC$ interaction is estimated with contrast weights $\frac{1}{4}\left(
1,-1,-1,1,-1,1,1,-1\right) $. \ No one of these 7 effects is aliased with
another. \ The $\frac{1}{2}$-fraction still has contrasts for the three main
effects; for instance, it has contrast $\frac{1}{2}\left(
y_{1}+y_{4}-y_{6}-y_{7}\right) $ with weights $\frac{1}{2}\left(
1,1,-1,-1\right) $ for $A$. \ However, the contrast for the main effect of $%
C $ in the $\frac{1}{2}$-fractional factorial equals the contrast for the $%
AB $ interaction; that is, both of these two contrasts equals $\frac{1}{2}%
\left( 1,-1,-1,1\right) $ where the $C$ main effect is read from the $C$
column of the right side of Table 1, and the $AB$ contrast is the product of
its $A$ and $B$ columns. \ We can estimate $\frac{1}{2}\mathrm{E}\left(
y_{1}-y_{4}-y_{6}+y_{7}\right) $ from the fractional factorial, but we
cannot tell whether we have estimated a main effect of $C$ or an interaction
of $A$ and $B$ or some combination of the two; so, $C$ is aliased with $AB$.

There is an elegant theory of aliasing in fractional factorials; see, for
instance, \cite{mukerjee2007modern}. \ In thinking about aliasing in
observational studies, three aspects of the theory of fractional factorials
deserve mention, namely contrasts versus models, unbalanced replication, and
covariates.
See, for example, \cite{blackwell2023noncompliance}, \cite{branson2016improving}, and \cite{dasgupta2015causal} for recent causal analyses of factorial designs.

A strength of the modern theory of fractional factorial designs is that it
is framed in terms of contrasts or weighted combinations of observations and
their ambiguities, not in terms of absolute faith in some linear model
picked a priori. \ Consequently, the alias structure of a fractional
factorial refers to specific ambiguities of specific contrasts, not to
whether a model is identified as a whole. \ In a fractional factorial larger
than Table 1 --- specifically, in a resolution IV design --- main effects
may be free of aliasing with 2-factor interactions, yet 2-factor
interactions may be aliased with each other; so, a model with all main
effects and 2-factor interactions is not identified, yet main effects are
identified, meaning that the model is interpretable on several levels, even
though it is not identified as a model. \ This situation arises in
observational studies also, even though it is rarely recognized as such when
it occurs.

The aliasing of contrasts that occurs in observational studies is forced
upon the investigator by the way some law or policy makes individuals
eligible or ineligible for treatment; so, the aliasing may lack the simple
and immediate form associated with some orthogonal experimental designs. \
Helpful in this regard is paragraph 4a.2(iv) in \cite{rao1973linear}: it
gives a necessary and sufficient condition for the parameter of a linear
model to be estimable --- i.e., not aliased --- and then gives the least
squares estimator as a contrast in the observations, even when the design
matrix may include other parameters or contrasts that are not estimable or
aliased. \ By adding columns to a design matrix and doing one step of linear
algebra, the alias structure and contrast are calculated. \ In observational
studies, we are working out the alias structure that has been forced upon
us, not choosing one alias structure in preference to another.

Fractional factorials are usually run to reduce the size of a costly
experiment, so it would be extremely uncommon to replicate the half fraction
on the right in Table 1, rather than run and perhaps replicate the full
fraction on the left in Table 1. \ In observational studies, aliasing is
typically created by public policies that define who is eligible for the
treatment; so, large, heavily replicated studies with basic aliasing are
common. \ If the two designs in Table 1 had treatment patterns that were
replicated, perhaps in an unbalanced allocation to treatment combinations,
then the contrast weights that attach to individual observations would
change (\cite{chattopadhyay2023implied}), but in the half fraction the
aliasing of the $C$ main effect with the $AB$ interaction would remain.

Adjustment for measured covariates plays a large role in observational
studies, and a smaller role in most fractional factorial designs. \ To
remove bias, an observational study must adjust for measured covariates that
predict both treatment assignment and outcomes. \ Typically, these
covariates are associated with treatment but do not rigidly determine
treatment; so, they are usually not aliased with treatment. \ Effect
modification is an interaction between a treatment and a covariate: the
magnitude or stability of the treatment effect varies predictably with the
level of the covariate. \ If a covariate $x_{k}$ is not aliased with any
effect, then effect modification between $x_{k}$ and effects is aliased in
the same way that the effects themselves are aliased.

\subsection{Aliasing in the Simplest Observational Study}

The simplest nontrivial form of aliasing in an observational study occurs
when a single policy or treatment is offered to individuals who meet certain
eligibility requirements after a certain start date, but the treatment is
not offered to similar individuals who would have met the same eligibility
requirements before the start date, and is not offered to dissimilar
ineligible individuals before or after the start date. \ There are four
cells defined by time and eligibility requirements, where only the
eligible-after cell is given treatment, with the design matrix in Table 2,
where $\mathbf{h}$ gives the weights for a treatment
effect estimate, here the least squares weights in a balanced design, which are the usual so-called \textquotedblleft
difference-in-differences\textquotedblright\ weights. \ In this case, the
main-effect contrasts for time, eligibility and treatment are not aliased,
but addition of a the time-by-eligibility interaction contrast introduces
aliasing with the treatment contrast estimate, $\mathbf{h}$. \ Various ways of coding this design matrix yield the same difference-in-difference contrast $\mathbf{h}$. \ This very
basic design has several flavors and a wide variety of names, including the
before-and-after nonequivalent control group design and the method of
difference-in-differences; see, for instance, \cite{cook2002experimental}, 
\cite{meyer1995natural}, \cite{shaikh2021randomization} and \cite%
{ye2023negative}. \ Studies using this approach include, for example, \cite%
{card2000minimum} and \cite{meyer1995workers}.

The example used here from \cite{lalive2006changes} is often described as a
\textquotedblleft difference-in-differences study,\textquotedblright\ but in
fact it has a more complex, more robust alias structure.

\section{An Observational Study of Unemployment Benefits}
\vspace{-.5cm}
\subsection{Treatments Aliased with Eligibility $\times $ Time Interactions}

In August of 1989, Austria made changes to its unemployment benefits,
increasing the monetary benefit amount (the replacement rate) for workers
whose previous jobs had low earnings (LE), and increasing the duration of
benefits for older workers with a long period of no or infrequent
unemployment (IU). \ An increase in benefits duration is denoted B and no
increase is b, while an increase in the replacement rate is R and no
increase is r. \ \cite{lalive2006changes} studied the effects of these
changes in unemployment benefits on the duration of unemployment. \ They
studied people who became unemployed in the two years before August 1989 and
other people who became unemployed in the two years after August 1989. \ A
textbook presentation of their study is given by 
\citet[Ch.~5]{cahuc2014labor}%
.

We focus on the subset of individuals in their study who were aged 40-55, as
no individual under 40 was eligible for a lengthened duration of
unemployment benefits. \ The categories, LE and IU, exist before and after
August 1989, but they determine B-or-b and R-or-r only after August 1989. \
After August 1989, IU determines B-or-b, and LE determines R-or-r. \ Unlike
a conventional difference in difference design, an interaction between Low
Earnings (LE) and Time does not lead to aliasing of the B-or-b effect, and
an interaction between Infrequent Unemployment (IU) and Time does not lead
to aliasing of the R-or-r effect.

There are four treatment types:
\vspace{-.5em}
\begin{itemize}
\setlength\itemsep{-.5em}
\item $\mathrm{BR}$ for an increase in benefit amount and duration,

\item $\mathrm{Br}$ for an increase in benefit duration only,

\item $\mathrm{bR}$ for an increase in benefit amount only,

\item $\mathrm{br}$ for no increase in benefits.
\end{itemize}

Before August 1989, these four eligibility categories exist, but the changes
had not gone into effect; so the categories are denoted $\overline{\mathrm{BR%
}}$, $\overline{\mathrm{Br}}$, $\overline{\mathrm{bR}}$, $\overline{\mathrm{%
br}}$ in the period before August 1989. \ The alias structure is determined
by the design matrix in Table 3.

Although the eligibility categories, LE and IU, exist before August 1989
when TIME $=-1$, the treatments B/b and R/r were not in operation, so Bb and
Rr are 0. \ Because Table 3 is of full column rank, all of the linear model
parameters for this design matrix Table 3 are all jointly estimable \cite[%
4a.2(iii)]{rao1973linear}, as is true in a conventional
difference-in-differences analysis. \ The least squares contrast weights $%
\mathbf{h}$ in Table 3 that estimate the $B/b$ effect are the usual
difference-in-difference contrast for $B/b$.

If we add the $\mathrm{LE\times TIME}$ interaction as a new column to the
design matrix, then $\mathrm{LE\times TIME=}\left( 2\times \mathrm{R/r}%
\right) \mathrm{-LE}$; so, there is a linear dependence and the model is no
longer identified (no longer of full column rank). \ However, the $\mathrm{%
B/b}$ effect is still estimable using the same contrast $\mathbf{h}$. \
Although adding an $\mathrm{LE\times TIME}$ interaction results in an
unidentified model with some aliasing, this added aliasing does not involve
the contrast $\mathbf{h}$ for the main effect of benefits duration.

In the previous argument, the roles of $\mathrm{B/b}$ and $\mathrm{R/r}$ can
be reversed if $\mathrm{IU\times TIME}$ is added to Table 3 instead of $%
\mathrm{LE\times TIME}$.

Several further subtleties are shown in Web-appendix Table 1. The benefits
available prior to August 1989 depended upon whether one is less than 50
years old and whether one worked for at least 3 of the prior 5 years. \
These two covariates do not determine eligibility for a benefits change in
August 1989 --- for instance, in every group $g$ in Table 3, there are
people under and over 50 years old --- but the two covariates affect the
magnitude of the change from before August 1989, so they are possible effect
modifiers. \ It is, therefore, important to balance these two covariates; see \S\ref{ssMoreComplexA}.

\subsection{Eight Treatments in Blocks of Size 4}

As seen in Table 4, the blocked study contains 16800 individuals in 4200
blocks of size 4, where $16800=4200\times 4$. \ The matching that formed the
blocks of size 4 is described in \S \ref{secMatching}. There are 6 types of
blocks, with 700 blocks of each type, or $2800=700\times 4$ individuals in
each block type. \ Each block contains two eligibility categories, before
and after two treatments were given to those two eligibility categories, and
consequently two of the four treatments, BR, Br, bR, and br. \ There are six
block types because $\binom{4}{2}=6$. \ This design is a particular type of
a confounded factorial in incomplete blocks \cite[\S 4.15]{wu2011experiments}%
.

The difference-in-difference contrast estimates a different comparison of
treatments in each block type, and the aliasing is different in each block
type. \ For example, block type 2 contains BR and bR, so it provides 700
within-block difference-in-differences that each estimate the effect of the
change in benefits duration, B-vs-b, at the high level R of the change in
the replacement rate R/r. \ In contrast, block 5 contains Br and br, so it
provides 700 difference-in-differences that estimate the effect of the
change in benefits duration, B-vs-b, at the low level r of the change in the
replacement rate R/r.

\section{Basic Theory of Aliasing in Observational Studies}
\vspace{-.5cm}
\subsection{Ignorable Treatment Assignment with Aliasing}

There are $G$ treatment groups, $g=1,\ldots ,G$, where $\mathbf{Z}=\left(
Z_{1},\ldots ,Z_{G}\right) $ indicates the group to which an individual is
assigned, with $Z_{g}=1$ if the individual is in group $g$, and $Z_{g}=0$
otherwise. \ Because each individual is assigned to one group, $%
1=\sum_{g=1}^{G}Z_{g}$. \ In Table 3, $G=8$. \ An individual has observed
unaliased covariates $\mathbf{x}$ and aliased eligibility covariates $%
\mathbf{w}$. \ In Table 3, $\mathbf{w}=\left( w^{\prime },w^{\prime \prime
},w^{\prime \prime \prime }\right) $ where $w^{\prime }$ is low income (LE), 
$w^{\prime \prime }$ is infrequent unemployment (IU) and $w^{\prime \prime
\prime }$ is time, before or after the change in policy. \ Group membership $%
\mathbf{Z}$ is determined by $\mathbf{w}$, and conversely everyone with $%
Z_{g}=1$ has the same value, say $\mathbf{w}_{g}$, of the eligibility
covariates; i.e., $\Pr \left( \left. Z_{g}=1\,\right\vert \,\mathbf{x},\,%
\mathbf{w=w}_{g}\right) =1$ and $\Pr \left( \left. Z_{g}=1\,\right\vert \,%
\mathbf{x},\,\mathbf{w\neq w}_{g}\right) =0$ for $g=1,\ldots ,G$.

An individual has response $r_{g}$ if given treatment $g$, but the effect $%
r_{g}-r_{g^{\prime }}$ caused by receiving treatment $g$ in lieu of
treatment $g^{\prime }$ is not observed for any individual. \ The observed
response is $R=\sum_{g=1}^{G}Z_{g}\,r_{g}$.

\begin{definition}
\label{DefIwithA} Treatment assignment is ignorable given covariates $%
\mathbf{x}$ with aliasing by functions $\beta \left( \mathbf{x},\,\mathbf{w}%
\right) $ if\begin{equation}
\mathbf{Z}\left. \,\indep\,\left\{
r_{1}-\beta \left( \mathbf{x},\,\mathbf{w}_{1}\right) ,\,\ldots
,\,r_{G}-\beta \left( \mathbf{x},\,\mathbf{w}_{G}\right) \right\}
\,\,\right\vert \,\mathbf{x}  \label{eqIwithAci}
\end{equation}%
and%
\begin{equation}
0<\Pr \left( \left. Z_{g}=1\,\right\vert \,\mathbf{x}\right) <1\text{ for }%
g=1,\ldots ,G\text{.}  \label{eqIwithApositive}
\end{equation}\end{definition}

In Definition \ref{DefIwithA}, $\Pr \left( \left. Z_{g}=1\,\right\vert \,%
\mathbf{x},\,\mathbf{w}\right) $ is zero or one, but (\ref{eqIwithApositive}%
) speaks about $\Pr \left( \left. Z_{g}=1\,\right\vert \,\mathbf{x}\right) $%
, not $\Pr \left( \left. Z_{g}=1\,\right\vert \,\mathbf{x},\,\mathbf{w}%
\right) $.

If (\ref{eqIwithApositive}) is true, then the function $\mathrm{E}\left(
\left. R\,\right\vert \,Z_{g}=1,\,\mathbf{x}\right) $ is estimable as it
involves only quantities that are jointly observed: it is the regression of $%
R$ on $\mathbf{x}$\ in the nonempty subpopulation with $Z_{g}=1$. \ If $%
0=\beta \left( \mathbf{x},\,\mathbf{w}_{1}\right) =\cdots =\beta \left( 
\mathbf{x},\,\mathbf{w}_{G}\right) $, then Definition \ref{DefIwithA} is the
same as ignorable treatment assignment given $\mathbf{x}$ alone in \cite%
{rosenbaum2023propensity}; so, in this special case the
expectation $\mathrm{E}\left( \left. r_{g}-r_{g^{\prime }}\,\right\vert \,%
\mathbf{x}\right) $ of the causal effect, $r_{g}-r_{g^{\prime }}$, can be
estimated from observed quantities,%
\begin{eqnarray}
\mathrm{E}\left( \left. r_{g}-r_{g^{\prime }}\,\right\vert \,\mathbf{x}%
\right) &=&\mathrm{E}\left( \left. r_{g}\,\right\vert \,Z_{g}=1,\,\mathbf{x}%
\right) -\mathrm{E}\left( \left. r_{g^{\prime }}\,\right\vert \,Z_{g^{\prime
}}=1,\,\mathbf{x}\right)  \label{eqExpectedggprimeEffect} \\
&=&\mathrm{E}\left( \left. R\,\right\vert \,Z_{g}=1,\,\mathbf{x}\right) -%
\mathrm{E}\left( \left. R\,\right\vert \,Z_{g^{\prime }}=1,\,\mathbf{x}%
\right) \text{.}  \notag
\end{eqnarray}%
Of course, (\ref{eqExpectedggprimeEffect}) used $0=\beta \left( \mathbf{x},\,%
\mathbf{w}_{1}\right) =\cdots =\beta \left( \mathbf{x},\,\mathbf{w}%
_{G}\right) $ and does not otherwise follow from (\ref{eqIwithAci}). \ In
brief, Definition \ref{DefIwithA} states a weaker condition than ignorable
treatment assignment given $\mathbf{x}$, in the sense that the value of $%
\mathbf{Z}$ is a 1-to-1 function of $\mathbf{w}$. \ It encompasses forms of equi-confounding in the assignment mechanism, such as the common assumption of parallel trends in difference-in-differences, among other examples \citep{ding2019bracketing}.

Let $h_{g}$ be specific constants where $0=\sum_{g=1}^{G}h_{g}$ and at least
one $h_{g}$ is not zero. \ \cite{mukerjee2018using} discuss contrasts for
causal effects in randomized experiments. \ Suppose that we wish to estimate
the expectation of a specific causal contrast,$\,\mathrm{E}\left( \left.
\sum_{g=1}^{G}h_{g}\,r_{g}\,\right\vert \,\mathbf{x}\right)
=\sum_{g=1}^{G}h_{g}\,\mathrm{E}\left( \left. r_{g}\,\right\vert \,\mathbf{x}%
\right) $. \ In (\ref{eqExpectedggprimeEffect}), for instance, $\mathrm{E}%
\left( \left. r_{g}-r_{g^{\prime }}\,\right\vert \,\mathbf{x}\right)
=\sum_{g=1}^{G}h_{g}\,\mathrm{E}\left( \left. r_{g}\,\right\vert \,\mathbf{x}%
\right) $ with $h_{g}=1=-h_{g^{\prime }}$ and $h_{g^{\prime \prime }}=0$ for 
$g^{\prime \prime }\notin \left\{ g,\,g^{\prime }\right\} $. \
Alternatively, $h_{g}$ might be given by $\mathbf{h}=\left( h_{1},\ldots
,h_{8}\right) $ in Table 3.

\begin{definition}
\label{DefNotAliasedContrast} Contrast $h_{g}$, $g=1,\ldots ,G$, is not
aliased with $\beta \left( \mathbf{x},\,\mathbf{w}\right) $ if $%
\sum_{g=1}^{G}h_{g}\,\beta \left( \mathbf{x},\,\mathbf{w}_{g}\right) =0$ for
all $\mathbf{x}$.
\end{definition}

Typically, as in a fractional factorial experiment, an assumption that $%
\sum_{g=1}^{G}h_{g}\,\beta \left( \mathbf{x},\,\mathbf{w}_{g}\right) =0$ is
based on an assumption that $\beta \left( \mathbf{x},\,\mathbf{w}\right) $
has a particular simple form, not by assuming a specific function $\beta
\left( \mathbf{x},\,\mathbf{w}\right) $. \ For example, define $h_{g}$ by $%
\mathbf{h}=\left( h_{1},\ldots ,h_{8}\right) $ in Table 3. \ For $\mathbf{w}%
=\left( w^{\prime },w^{\prime \prime },w^{\prime \prime \prime }\right) $ in
Table 3 and 
\begin{equation}
\beta \left( \mathbf{x},\,\mathbf{w}\right) =\xi \left( \mathbf{x}%
,\,w^{\prime },\,w^{\prime \prime }\right) +\eta \left( \mathbf{x}%
,\,w^{\prime },\,w^{\prime \prime \prime }\right)   \label{eqSpecialAlias}
\end{equation}%
for any functions $\xi \left( \mathbf{x},\,w^{\prime },\,w^{\prime \prime
}\right) $ and $\eta \left( \mathbf{x},\,w^{\prime },\,w^{\prime \prime
\prime }\right) $, then $0=\sum_{g=1}^{G}h_{g}\,\beta \left( \mathbf{x},\,%
\mathbf{w}_{g}\right) $ in Table 3\ and contrast $h_{g}$, $g=1,\ldots ,G$,
is not aliased with $\beta \left( \mathbf{x},\,\mathbf{w}\right) $. \ For
instance, the contrast used to estimate the main effect of benefit duration is
not aliased with the main effect of infrequent unemployment (IU or $%
w^{\prime \prime }$), nor with the main effects and two-factor interaction
of low earnings (LE or $w^{\prime }$) and time (or $w^{\prime \prime \prime }
$). \ Note that (\ref{eqSpecialAlias}) permits any possible interaction
between $\mathbf{x}$, $w^{\prime }$ and$\,w^{\prime \prime }$ in $\xi \left( 
\mathbf{x},\,w^{\prime },\,w^{\prime \prime }\right) $, and any possible
interaction  $\mathbf{x},\,w^{\prime }$, and$\,w^{\prime \prime \prime }$ in 
$\eta \left( \mathbf{x},\,w^{\prime },\,w^{\prime \prime \prime }\right) $;
however, it precludes an interaction between, for example, $w^{\prime \prime
}$ and $w^{\prime \prime \prime }$. \ Because $R/r$ is determined by $\left(
w^{\prime },\,w^{\prime \prime \prime }\right) $, the same contrast $\mathbf{%
h}$ in Table 3 is not aliased with (R/r, LE and Time) jointly ---
essentially, R/r is a function of $\left( w^{\prime },\,w^{\prime \prime
\prime }\right) $ and so is redundant in $\eta \left( \mathbf{x},\,w^{\prime
},\,w^{\prime \prime \prime }\right) $.

\begin{proposition}
\label{propEstimateContrast} If treatment assignment is ignorable given
covariates $\mathbf{x}$ with aliasing by functions $\beta \left( \mathbf{x}%
,\,\mathbf{w}\right) $, and if contrast $h_{g}$, $g=1,\ldots ,G$, is not
aliased with $\beta \left( \mathbf{x},\,\mathbf{w}\right) $, then the
expected causal contrast $\mathrm{E}\left( \left.
\sum_{g=1}^{G}h_{g}\,r_{g}\,\right\vert \,\mathbf{x}\right) $ equals a
contrast of estimable functions
\begin{equation*}
\mathrm{E}\left( \left. \sum_{g=1}^{G}h_{g}\,r_{g}\,\right\vert \,\mathbf{x}%
\right) =\sum_{g=1}^{G}h_{g}\,\mathrm{E}\left( \left. R\,\right\vert
\,Z_{g}=1,\,\mathbf{x}\right) \text{.}
\end{equation*}
\end{proposition}

\begin{proof}
Because (\ref{eqIwithApositive}) holds, $\mathrm{E}\left( \left.
R\,\right\vert \,Z_{g}=1,\,\mathbf{x}\right) $ is estimable from observable
joint distributions for $g=1,\ldots ,\,G$. \ Repeatedly using the equality $%
0=\sum_{g=1}^{G}h_{g}\,\beta \left( \mathbf{x},\,\mathbf{w}_{g}\right) $, we
have:%
\begin{equation*}
\sum_{g=1}^{G}h_{g}\,\mathrm{E}\left( \left. R\,\right\vert \,Z_{g}=1,\,%
\mathbf{x}\right) =\sum_{g=1}^{G}h_{g}\,\mathrm{E}\left( \left.
r_{g}\,\right\vert \,Z_{g}=1,\,\mathbf{x}\right) \text{, as }R=r_{g}\text{
if }Z_{g}=1
\end{equation*}%
\begin{equation*}
=\sum_{g=1}^{G}h_{g}\,\mathrm{E}\left( \left. r_{g}\,\right\vert \,Z_{g}=1,\,%
\mathbf{x}\right) -\sum_{g=1}^{G}h_{g}\,\beta \left( \mathbf{x},\,\mathbf{w}%
_{g}\right) \text{ as }0=\sum_{g=1}^{G}h_{g}\,\beta \left( \mathbf{x},\,%
\mathbf{w}_{g}\right)
\end{equation*}%
\begin{equation*}
=\sum_{g=1}^{G}h_{g}\,\mathrm{E}\left\{ \left. r_{g}-\beta \left( \mathbf{x}%
,\,\mathbf{w}_{g}\right) \,\right\vert \,Z_{g}=1,\,\mathbf{x}\right\} \text{
as }\mathbf{w}=\mathbf{w}_{g}\text{ iff }Z_{g}=1
\end{equation*}%
\begin{eqnarray*}
&=&\sum_{g=1}^{G}h_{g}\,\mathrm{E}\left\{ \left. r_{g}-\beta \left( \mathbf{x%
},\,\mathbf{w}_{g}\right) \,\right\vert \,\mathbf{x}\right\} \text{ by (\ref%
{eqIwithAci})} \\
&=&\sum_{g=1}^{G}h_{g}\,\mathrm{E}\left( \left. r_{g}\,\right\vert \,\mathbf{%
x}\right) -\sum_{g=1}^{G}h_{g}\,\mathrm{E}\left\{ \left. \beta \left( 
\mathbf{x},\,\mathbf{w}_{g}\right) \,\right\vert \,\mathbf{x}\right\} \\
&=&\sum_{g=1}^{G}h_{g}\,\mathrm{E}\left( \left. r_{g}\,\right\vert \,\mathbf{%
x}\right) -\mathrm{E}\left\{ \left. \sum_{g=1}^{G}h_{g}\,\beta \left( 
\mathbf{x},\,\mathbf{w}_{g}\right) \,\right\vert \,\mathbf{x}\right\} \\
&=&\sum_{g=1}^{G}h_{g}\,\mathrm{E}\left( \left. r_{g}\,\right\vert \,\mathbf{%
x}\right) \text{ as }0=\sum_{g=1}^{G}h_{g}\,\beta \left( \mathbf{x},\,%
\mathbf{w}_{g}\right) \text{.}
\end{eqnarray*}
\end{proof}

\subsection{Propensity Scores and Balancing Scores}

Much of the basic theory of ignorable treatment assignment and propensity
scores (\cite{rosenbaum2023propensity}) carries over to
the situation with aliasing, even when treatment assignment is a
deterministic function of the aliased covariates, $\mathbf{w}$, as we
briefly demonstrate. \ The practical point of Proposition \ref%
{propEstimateContrastPropensity} is that Proposition \ref%
{propEstimateContrast} is useful without conditioning on all $\mathbf{x}$;
rather, it suffices to condition or match on enough to balance $\mathbf{x}$.

Define $\mathbf{e}\left( \mathbf{x}\right) =\left\{ \Pr \left( \left.
Z_{1}=1\,\right\vert \,\mathbf{x}\right) ,\ldots ,\Pr \left( \left.
Z_{G}=1\,\right\vert \,\mathbf{x}\right) \right\} =\left\{ e_{1}\left( 
\mathbf{x}\right) ,\ldots ,e_{G}\left( \mathbf{x}\right) \right\} $. \
Proposition \ref{propBalancing}\ is due to \cite{imai2004causal} and does
not involve aliasing.

\begin{proposition}
\label{propBalancing} (Imai and Van Dyk) \ (i) $\ \mathbf{Z}\left. \,%
\indep\,\mathbf{x}\,\,\right\vert \,%
\mathbf{e}\left( \mathbf{x}\right) $ and 
\begin{equation}
\text{(ii) }\ \mathbf{Z}\left. \,\indep\,%
\mathbf{x}\,\,\right\vert \,\left\{ \mathbf{e}\left( \mathbf{x}\right) ,\,%
\mathbf{f}\left( \mathbf{x}\right) \right\} \text{ for any }\mathbf{f}\left(
\cdot \right) \text{.}  \label{eqBalance}
\end{equation}
\end{proposition}

\begin{proof}
As $\mathbf{e}\left( \mathbf{x}\right) $ is a function of $\mathbf{x}$, to
demonstrate (i) the task is to show $\Pr \left( \left. \mathbf{Z}%
\,\right\vert \,\mathbf{x}\right) =\Pr \left\{ \left. \mathbf{Z}%
\,\right\vert \,\mathbf{e}\left( \mathbf{x}\right) \right\} $, which follows
from: 
\begin{eqnarray*}
\Pr \left\{ \left. Z_{g}=1\,\right\vert \,\mathbf{e}\left( \mathbf{x}\right)
\right\} &=&\mathrm{E}\left[ \Pr \left\{ \left. Z_{g}=1\,\right\vert \,%
\mathbf{x}\right\} \,|\,\mathbf{e}\left( \mathbf{x}\right) \right] =\mathrm{E%
}\left\{ e_{g}\left( \mathbf{x}\right) \,|\,\mathbf{e}\left( \mathbf{x}%
\right) \right\} \\
&=&e_{g}\left( \mathbf{x}\right) =\Pr \left( \left. Z_{g}=1\,\right\vert \,%
\mathbf{x}\right) \text{.}
\end{eqnarray*}%
Also, (ii) follows from (i) by Lemma 3 of \cite{dawid1979conditional}.
\end{proof}

Although Proposition \ref{propEstimateContrast} gives a sufficient condition
for the estimability of a causal contrast, $\mathrm{E}\left( \left.
\sum_{g=1}^{G}h_{g}\,r_{g}\,\right\vert \,\mathbf{x}\right) $, a weaker
condition is given in Proposition \ref{propEstimateContrastPropensity}. \
Instead of matching or otherwise adjusting for all of $\mathbf{x}$, it
suffices to adjust for a function of $\mathbf{x}$, specifically for any
function that is at least as fine as the propensity score, that is for any $%
\left\{ \mathbf{e}\left( \mathbf{x}\right) ,\,\mathbf{f}\left( \mathbf{x}%
\right) \right\} $. \ Taking $\mathbf{f}\left( \mathbf{x}\right) =\mathbf{x}$
in Proposition \ref{propEstimateContrastPropensity} yields Proposition \ref%
{propEstimateContrast}, whereas if $\mathbf{f}\left( \mathbf{x}\right) $ is
set to a constant that ignores $\mathbf{x}$, then Proposition \ref%
{propEstimateContrastPropensity} says it suffices to match or block for $%
\mathbf{e}\left( \mathbf{x}\right) $ alone. \ Proposition \ref%
{propEstimateContrastPropensity} says that we need not match groups exactly
for $\mathbf{x}$, providing $\mathbf{x}$ is balanced across groups in
matched blocks in the sense that (\ref{eqBalance}) holds.

\begin{proposition}
\label{propEstimateContrastPropensity} If treatment assignment is ignorable
given covariates $\mathbf{x}$ with aliasing by functions $\beta \left( 
\mathbf{x},\,\mathbf{w}_{g}\right) $, then for any $\mathbf{f}\left( \cdot
\right) $ it is also ignorable given $\left\{ \mathbf{e}\left( \mathbf{x}%
\right) ,\,\mathbf{f}\left( \mathbf{x}\right) \right\} $ with aliasing by
functions $\beta \left( \mathbf{x},\,\mathbf{w}_{g}\right) $, i.e., 
\begin{equation}
\mathbf{Z}\left. \,\indep\,\left\{
r_{1}-\beta \left( \mathbf{x},\,\mathbf{w}_{1}\right) ,\,\ldots
,\,r_{G}-\beta \left( \mathbf{x},\,\mathbf{w}_{G}\right) \right\}
\,\,\right\vert \,\left\{ \mathbf{e}\left( \mathbf{x}\right) ,\,\mathbf{f}%
\left( \mathbf{x}\right) \right\} \text{,}  \label{eqIwithAciPropensity}
\end{equation}%
\begin{equation}
0<\Pr \left\{ \left. Z_{g}=1\,\right\vert \,\mathbf{e}\left( \mathbf{x}%
\right) ,\,\mathbf{f}\left( \mathbf{x}\right) \right\} <1\text{ for }%
g=1,\ldots ,G\text{;}  \label{eqIwithApositivePropensity}
\end{equation}%
so, the expected causal contrast $\mathrm{E}\left\{ \left.
\sum_{g=1}^{G}h_{g}\,r_{g}\,\right\vert \,\mathbf{e}\left( \mathbf{x}\right)
,\,\mathbf{f}\left( \mathbf{x}\right) \right\} $ equals a contrast of
estimable functions%
\begin{equation}
\mathrm{E}\left\{ \left. \sum_{g=1}^{G}h_{g}\,r_{g}\,\right\vert \,\mathbf{e}%
\left( \mathbf{x}\right) ,\,\mathbf{f}\left( \mathbf{x}\right) \right\}
=\sum_{g=1}^{G}h_{g}\,\mathrm{E}\left\{ \left. R\,\right\vert \,Z_{g}=1,\,%
\mathbf{e}\left( \mathbf{x}\right) ,\,\mathbf{f}\left( \mathbf{x}\right)
\right\} \text{.}  \label{eqEstimatePropContrast}
\end{equation}
\end{proposition}

\begin{proof}
Write $\mathbf{d}=\left\{ r_{1}-\beta \left( \mathbf{x},\,\mathbf{w}%
_{1}\right) ,\,\ldots ,\,r_{G}-\beta \left( \mathbf{x},\,\mathbf{w}%
_{G}\right) \right\} $ and $\mathbf{s}\left( \mathbf{x}\right) =\left\{ 
\mathbf{e}\left( \mathbf{x}\right) ,\,\mathbf{f}\left( \mathbf{x}\right)
\right\} $. \ Then to demonstrate (\ref{eqIwithAciPropensity}), it must be
shown that $\Pr \left\{ \left. \mathbf{Z}\,\right\vert \,\mathbf{d},\,%
\mathbf{s}\left( \mathbf{x}\right) \right\} =\Pr \left\{ \left. \mathbf{Z}%
\,\right\vert \,\mathbf{s}\left( \mathbf{x}\right) \right\} $. \ From (\ref%
{eqIwithAci}), $\Pr \left( \left. Z_{g}=1\,\right\vert \,\mathbf{d},\,%
\mathbf{x}\right) =\Pr \left( \left. Z_{g}=1\,\right\vert \,\mathbf{x}%
\right) $, so that%
\begin{gather*}
\Pr \left\{ \left. Z_{g}=1\,\right\vert \,\mathbf{d},\,\mathbf{s}\left( 
\mathbf{x}\right) \right\} =\mathrm{E}\left\{ \left. \Pr \left\{ \left.
Z_{g}=1\,\right\vert \,\mathbf{d},\,\mathbf{x}\right\} \,\right\vert \,%
\mathbf{d},\,\mathbf{s}\left( \mathbf{x}\right) \right\} \\
=\mathrm{E}\left\{ \left. \Pr \left( \left. Z_{g}=1\,\right\vert \,\mathbf{x}%
\right) \,\right\vert \,\mathbf{d},\,\mathbf{s}\left( \mathbf{x}\right)
\right\} =\mathrm{E}\left\{ \left. e_{g}\left( \mathbf{x}\right)
\,\right\vert \,\mathbf{d},\,\mathbf{s}\left( \mathbf{x}\right) \right\} 
\text{.}
\end{gather*}%
As $e_{g}\left( \mathbf{x}\right) $ is part of $\mathbf{s}\left( \mathbf{x}%
\right) =\left\{ \mathbf{e}\left( \mathbf{x}\right) ,\,\mathbf{f}\left( 
\mathbf{x}\right) \right\} $, 
\begin{gather*}
\mathrm{E}\left\{ \left. e_{g}\left( \mathbf{x}\right) \,\right\vert \,%
\mathbf{d},\,\mathbf{s}\left( \mathbf{x}\right) \right\} =e_{g}\left( 
\mathbf{x}\right) =\mathrm{E}\left\{ \left. e_{g}\left( \mathbf{x}\right)
\,\right\vert \,\mathbf{s}\left( \mathbf{x}\right) \right\} \\
=\mathrm{E}\left\{ \left. \Pr \left( \left. Z_{g}=1\,\right\vert \,\mathbf{x}%
\right) \,\right\vert \,\mathbf{s}\left( \mathbf{x}\right) \right\} =\Pr
\left\{ \left. Z_{g}=1\,\right\vert \,\mathbf{s}\left( \mathbf{x}\right)
\right\} \text{,}
\end{gather*}%
proving (\ref{eqIwithAciPropensity}). \ Then (\ref%
{eqIwithApositivePropensity}) follows from (\ref{eqIwithApositive}) because 
\begin{equation*}
\Pr \left\{ \left. Z_{g}=1\,\right\vert \,\mathbf{s}\left( \mathbf{x}\right)
\right\} =\mathrm{E}\left\{ \left. \Pr \left( \left. Z_{g}=1\,\right\vert \,%
\mathbf{x}\right) \,\right\vert \,\mathbf{s}\left( \mathbf{x}\right)
\right\} \text{.}
\end{equation*}%
Finally, (\ref{eqEstimatePropContrast}) follows from Proposition \ref%
{propEstimateContrast} with $\mathbf{x}$ replaced by $\left\{ \mathbf{e}%
\left( \mathbf{x}\right) ,\,\mathbf{f}\left( \mathbf{x}\right) \right\} $.
\end{proof}

\subsection{Checking Covariate Balance with Aliasing}

Blocking entails selecting individuals and grouping them into blocks that
balance $\mathbf{x}$ and are homogeneous in certain features $\mathbf{f}%
\left( \mathbf{x}\right) $ of $\mathbf{x}$. Let $b$ denote the block to
which an individual is assigned. \ Let $\mathcal{G}_{b}\subseteq \left\{
1,\ldots ,G\right\} $ be the groups present in block $b$. \ For instance, in
Table 4, the 700 blocks $b$ of type 2 all have $\mathcal{G}_{b}=\left\{
1,3,5,7\right\} $ in the first column of Table 3.

If blocks were homogeneous in $\left\{ \mathbf{e}\left( \mathbf{x}\right) ,\,%
\mathbf{f}\left( \mathbf{x}\right) \right\} $, then (\ref{eqBalance}) would
hold, and this motivates Definition \ref{DefBalance}. \ Recall that $A$ is
independent of $B$ if and only if $\mathrm{E}\left\{ \left. c\left( A\right)
\,\right\vert \,B\right\} =\mathrm{E}\left\{ c\left( A\right) \right\} $ for
all functions $c\left( \cdot \right) $. \ (More precisely: this condition
must hold for almost all $B$ and for all measurable $c\left( \cdot \right) $
for which the expectations exist, but we will not repeat this.)

\begin{definition}
\label{DefBalance} The distribution of $\mathbf{x}$ is balanced within
blocks if $\mathbf{x}\left. \,\indep\,%
\mathbf{Z}\,\,\right\vert \,b$, or equivalently if $\mathrm{E}\left\{ \left.
c\left( \mathbf{x}\right) \,\right\vert \,b,\,Z_{g}=1\right\} =$ $\mathrm{E}%
\left\{ \left. c\left( \mathbf{x}\right) \,\right\vert \,b\right\} $ for
every $c\left( \cdot \right) $ and $b$, or equivalently if for every $%
c\left( \cdot \right) $ and $b$, 
\begin{equation*}
\sum_{g\in \mathcal{G}_{b}}h_{g}\,\mathrm{E}\left\{ \left. c\left( \mathbf{x}%
\right) \,\right\vert \,b,\,Z_{g}=1\right\} =0\text{ for every }h_{g}\text{
such that }0=\sum_{g\in \mathcal{G}_{b}}h_{g}\text{.}
\end{equation*}
\end{definition}

A block design, with or without aliasing, must be checked for adequacy by
checking whether $\mathbf{x}$ is balanced by blocking; see, for instance, 
\cite{yu2021evaluating} for a modern discussion and example.

The situation is different for eligibility covariates, $\mathbf{w}$, that
are aliased with groups. \ Eligibility covariates $\mathbf{w}$ are always
imbalanced, because they determine $\mathbf{Z}$; so, they never satisfy
Definition \ref{DefBalance}. \ However, a function such as $\eta \left( 
\mathbf{x},\,w^{\prime },\,w^{\prime \prime  }\right) $ or $\xi \left( 
\mathbf{x},\,w^{\prime }, \, w^{\prime \prime \prime } \right) $ in (\ref{eqSpecialAlias}) may be
balanced in some blocks $b$ for a specific contrast such as $\mathbf{h}$ in
Table 3, and this motivates Definition \ref{DefBalance2}.

\begin{definition}
\label{DefBalance2} A function $\zeta \left( \mathbf{x},\mathbf{w}\right) $
of unaliased covariates $\mathbf{x}$\ and aliased covariates $\mathbf{w}$ is
balanced with respect to contrast $h_{g}$, $g=1,\ldots ,G$, in block $b$ if:%
\begin{equation}
0=\sum_{g\in \mathcal{G}_{b}}h_{g}\text{ \ and \ }\sum_{g\in \mathcal{G}%
_{b}}h_{g}\,\mathrm{E}\left\{ \left. \zeta \left( \mathbf{x},\mathbf{w}%
\right) \,\right\vert \,b,\,Z_{g}=1\right\} =0\text{.}
\label{eqBalanceContrast}
\end{equation}
\end{definition}

To illustrate, in Table 3, consider a block $b$ of size 4 containing
treatments $g=1$, 3, 5, 7, or $\overline{\mathrm{BR}}$, $\overline{\mathrm{bR%
}}$, $\mathrm{BR}$, $\mathrm{bR}$; so, $\mathcal{G}_{b}=\left\{
1,\,3,\,5,\,7\right\} $, and $\sum_{g\in \mathcal{G}_{b}}h_{g}=-\frac{1}{4}+%
\frac{1}{4}+\frac{1}{4}-\frac{1}{4}=0$. \ Suppose that the conditions of
Proposition \ref{propEstimateContrastPropensity} hold and the blocking
controlled $\left\{ \mathbf{e}\left( \mathbf{x}\right) ,\,\mathbf{f}\left( 
\mathbf{x}\right) \right\} $. \ Then that contrast applied to $R$ in this
block is a estimate of the effect of an increase in benefits duration, $B/b$%
, if $\beta \left( \mathbf{x},\,\mathbf{w}\right) $ has the form $\beta
\left( \mathbf{x},\,\mathbf{w}\right) =\eta \left( \mathbf{x},\,w^{\prime
},\,w^{\prime \prime }\right) +\xi \left( \mathbf{x},\,w^{\prime}, \, w^{\prime \prime \prime }\right) $ in (\ref{eqSpecialAlias}). \ So, Definition \ref%
{DefBalance2} suggests checking various functions $\zeta \left( \mathbf{x},%
\mathbf{w}\right) $ of this form, for instance $\zeta \left( \mathbf{x},%
\mathbf{w}\right) =x_{1}\times w^{\prime }\times w^{\prime \prime \prime }$,
to see whether these functions are balanced in the sense of (\ref%
{eqBalanceContrast}). \ Of course, a function $\eta \left( \mathbf{x}%
,\,w^{\prime },\,w^{\prime \prime \prime }\right) =w^{\prime }$ is
necessarily balanced in the sense of (\ref{eqBalanceContrast}), but a
function $\eta \left( \mathbf{x},\,w^{\prime },\,w^{\prime \prime \prime
}\right) =x_{1}\times w^{\prime }$ may or many not be balanced, and finding
an imbalance suggests that the blocking has not adequately balanced the
covariates. \ 

A block $\mathcal{G}_{b}=\left\{ 1,\,3,\,5,\,7\right\} $ provides a contrast
to estimate the effect of benefits duration, $B/b$, for LE or $w^{\prime }=1$%
. \ A block $\mathcal{G}_{b}=\left\{ 2,\,4,\,6,\,8\right\} $ provides a
contrast to estimate the effect of benefits duration, $B/b$, for LE or $%
w^{\prime }=-1$. \ Those estimates may have different estimands if there is
effect modification by LE.

\section{Benefits Duration and Weeks of Unemployment}
\vspace{-.5cm}
\subsection{Checking Covariate Balance}

Figure 1 checks the covariate age, $x_{3}$, for balance in each of the six
types of blocks and for $h_{g}=1$ or $h_{g}=-1$ in Table 4. \ Here,
covariate balance refers to a particular contrast $h_{g}$, as in Definition %
\ref{DefBalance2}. \ A difference-in-difference effect estimate in a block
type in Table 4 will compare individuals with $h_{g}=1$ to individuals with $%
h_{g}=-1$, and in each block type in Figure 1, age does look balanced for $%
h_{g}=1$ or $h_{g}=-1$. \ For an unaliased covariate like age in $\mathbf{x}$%
, balance could be checked in over four groups defined by $g\in \mathcal{G}%
_{b}$ --- that is balance could be checked for every contrast in Definition %
\ref{DefBalance}, rather than for the two groups defined by $h_{g}=1$ or $%
h_{g}=-1$ in Table 4, but this is not possible for aliased covariates.

The \textquotedblleft low earnings\textquotedblright\ covariate LE or $%
w^{\prime }$ in Table 3 is perfectly balanced over $h_{g}=1$ or $h_{g}=-1$
by design in every block type, 1 to 6: it is equally represented, block by
block, in after and before periods. \ So, there is nothing to check so far
as LE or $w^{\prime }$ is concerned. \ However, the LE-by-Time interaction, $%
w^{\prime }\times w^{\prime \prime \prime }$, is perfectly balanced over $%
h_{g}=1$ or $h_{g}=-1$ in blocks of types 2 = BR/bR and 5 = Br/br, but it is
aliased with $h_{g}=1$ or $h_{g}=-1$ in other blocks. \ In words, an
interaction between low earnings and time does not bias contrasts estimating
the effect of B-vs-b at high R in blocks of type 2, nor the effect of B-vs-b
at low r in blocks of type 5. In parallel, the interaction between
\textquotedblleft infrequent unemployment\textquotedblright\ and Time or $%
w^{\prime \prime }\times w^{\prime \prime \prime }$ is balanced over $%
h_{g}=1 $ or $h_{g}=-1$ in blocks of types 1 = BR/Br and 6 = bR/br, but is
aliased in other block types.

The situation is different for the interaction covariate \textquotedblleft
age $\times $ LE\textquotedblright\ or $x_{3}\times w^{\prime }$. \ Like age
itself, $x_{3}\times w^{\prime }$ may be balanced or not over $h_{g}=1$ or $%
h_{g}=-1$, and we need to look at the data to judge whether matching has
balanced $x_{3}\times w^{\prime }$ over $h_{g}=1$ or $h_{g}=-1$. This
covariate $x_{3}\times w^{\prime }$ is checked in Figure 2, consistent with
Definition \ref{DefBalance2}. \ Because $w^{\prime }=\pm 1$, both positive
and negative values of $x_{3}\times w^{\prime }$ may appear in Figure 2,
although only $w^{\prime }=1$ appears in block type 2 so only positive
values appear Figure 2, and only $w^{\prime }=-1$ appears in block type 5,
so only negative values appear in Figure 2. \ In every block, however,
individuals with positive and negative contrast values, $h_{g}=1$ or $%
h_{g}=-1$, have similar distributions of $x_{3}\times w^{\prime }$; so, an
imbalance in $x_{3}\times w^{\prime }$ cannot bias the contrast.

Table 5 is analogous to Figure 1. \ For each covariate $x_{k}$, Table 5
compares the distribution of $x_{k}$ for individuals with positive and
negative contrast values, $h_{g}=1$ or $h_{g}=-1$, to complete randomization
of $x_{k}$ to $h_{g}=1$ or $h_{g}=-1$. \ This is done separately in each of
the six block types, as $h_{g}=1$ or $h_{g}=-1$ has a different meaning in
different block types, and then the within-type $P$-values are combined over
block types using \cite{zaykin2002truncated}'s truncated product method of
combining independent $P$-values, as implemented in the \texttt{truncatedP}
method in the \texttt{R} package \texttt{sensitivitymv} with the default
truncation of 0.2. \ This method defines a new statistic as the product of all $P$-values less
than or equal to the truncation point, or 1 if no $P$-value is below the truncation point.  The null distribution of this new statistic is determined, and that yields its $P$-value, which is reported in column C of Table 5. \ If the truncation point were 1 instead of 0.2, then the
method would be equivalent to Fisher's method for combining independent $P$%
-values. \ As all of the $P$-values in Table 5 are above 0.2, the combined $%
P $-value is always 1. \ The balance of covariates in Table 5 is much better
than expected from random assignment to $h_{g}=1$ or $h_{g}=-1$.

Several aspects of Table 5 deserve note, as they serve to illustrate
Definition \ref{DefBalance2}. \ Recall from Definition \ref{DefBalance2}
that a covariate may be balanced with respect to contrast $h_{g}$ over
several groups even though some individual groups do not overlap with
respect to that covariate. \ The covariate $w^{\prime }$ or LE is perfectly
balanced because the same number of individuals in group $h_{g}=1$ have $w^{\prime
}=1$ as in
group $h_{g}=-1$. \ The same is true of the covariate $w^{\prime \prime }$
or IU, and of the time period $w^{\prime \prime \prime }$. \ Both $w^{\prime
}$ and $w^{\prime \prime }$ are two-level covariates formed by cutting
continuous variables, namely \textquotedblleft prior wage\textquotedblright\
for $w^{\prime }$ and \textquotedblleft relative
employment\textquotedblright\ for $w^{\prime \prime }$. \ Because
eligibility cuts continuous covariates, the support of the continuous
covariates does not overlap in some pairs of treatment groups. \
Nonetheless, these continuous covariates are balanced with respect to $h_{g}$
in the sense of Definition \ref{DefBalance2}: the contrast groups, $h_{g}=1$
or $h_{g}=-1$, are balanced with respect to the continuous covariates by
merging two groups with nonoverlapping support. \ An estimator of a
treatment effect is a contrast in outcomes that will compare outcomes in the 
$h_{g}=1$ or $h_{g}=-1$ groups, and that comparison is balanced with respect
to covariates.

Table 6 is similar to Table 5, except Table 6 resembles Figure 2 in
examining the balance of $x_{k}\times w^{\prime }$, rather than the balance
of $x_{k}$ in Table 5. In other words, Table 6 is checking the balance of an
interaction between an eligibility covariate, $w^{\prime }$, and a covariate 
$x_{k}$. \ An additional table for $x_{k}\times w^{\prime \prime }$ is in
the Web Appendix. \ In this table also, the balance is much better than is
expected from complete randomization of individuals to $h_{g}=1$ or $h_{g}=-1$. \ So,
covariates of the form $x_{k}\times w^{\prime }$ are balanced, and therefore are not likely to bias an estimate
that compares individuals with $h_{g}=1$ and $h_{g}=-1$.

Appendix Tables 2 to 5 are similar to Table 6, but they explore different interactions between covariates $x_{k}$ and eligibility covariates, $w^{\prime }$, $w^{\prime\prime }$, and $w^{\prime\prime\prime }$.  Appendix Tables 2 and 3 concern interactions that \textit{could} be balanced given the alias structure in Table 3, and all of those interactions are balanced.  Appendix Tables 4 and 5 concern interactions that \textit{cannot} be balanced given the alias structure in Table 3, and all of those interactions are severely imbalanced.  When certain covariate combinations perfectly determine the treatment, those covariate combinations cannot be balanced.  In the four tables of covariates or interactions that could be balanced, there are a total of 504 covariate/block-type situations; there, the median of 504 $P$-values is 0.8871 and the minimum is 0.2379, so the balance is much better than expected by complete randomization.  In common practice, a linear model would not include so many covariates and interactions; yet here, all 504 are better balanced than by complete randomization.  In the two tables where aliasing forces imbalances, there are 252 covariate/block-type interactions, with a median $P$-value of 0.0000 and an upper quartile of 0.5708.  The upper quartile of 0.5708 reflects the fact that comparisons in certain block types are not aliased by eligibility covariates that create aliasing in other block types.  Covariates and interactions that were not aliased were balanced by matching, but covariates that are aliased cannot be controlled by any statistical method.

\vspace{-.5cm}

\subsection{Examining the Outcome: Duration of Unemployment}

Figure 3 compares two direct estimates of the effect of the increase in
benefits duration, B-versus-b, from blocks of types 2 and 5. \ However,
block types 2 and 5 are different. \ In blocks of type 2, the replacement
rate was increased in the after period, R, while in blocks of type 5 the
replacement rate not increased in the after period. \ The two left panels of
Figure 3 show the duration of unemployment, capped at 104 weeks or two
years. \ The contrast weights, $h_{g}=1$ or $h_{g}=-1$, appear at the tops
of the boxplots. \ An increase in benefit duration B in the after period is
associated with a somewhat longer duration of unemployment in the after
period, whether the replacement rate is increased (type 2 blocks) or not
(type 5 blocks). \ The distributions have long tails, somewhat obscuring the
typical durations.

Each block produces one difference-in-difference, or 700
difference-in-differences from blocks of type 2 and 700 from blocks of type
5. \ These 1400 differences-in-differences are plotted in the final boxplot.
\ To see the typical difference clearly, the vertical axis has its tails
symmetrically transformed beyond $\pm \beta $ using a $p=-1$ transformation,
where $\beta $ is the 0.8 quantile of the absolute difference-in-differences
(\cite{rosenbaum2022new}). \ The median of the 1400 untransformed
differences-in-differences is 2.3 weeks, the quartiles are $-7.6$ and 12.3,
and the 10\% and 90\% points are $-19.4$ and 23.8.

The 1400 blocks compare durations of unemployment in two blocked groups, $%
h_{g}=1$ and $h_{g}=-1$. \ Compared to a random allocation of 2 of 4
individuals to $h_{g}=1$, the 1400 difference-in-differences are high, with
one-sided $P$-value of $1.1\times 10^{-16}$. \ The comparison is insensitive
to a bias that increases the odds of assignment to $h_{g}=1$ rather than $%
h_{g}=-1$ by a factor of $\Gamma =1.6$, as the upper bound on the $P$-value
is then 0.035. \ These calculations use the methods in \cite%
{rosenbaum2018sensitivity,rosenbaum2024bahadur} as implemented in the 
\texttt{gwgtRank} function of the \texttt{weightedRank} package in \texttt{R}
with the default settings. \ In a matched pair, a bias of $\Gamma =1.6$ is
equivalent to an unobserved covariate that doubles the odds of treatment and
increases the odds of a positive pair difference in outcome by 5-fold (\cite%
{rosenbaum2009amplification}).  This sensitivity analysis presumes that (\ref{eqIwithAciPropensity}) and (\ref{eqIwithApositivePropensity}) are possibly false, but would have been true if the conditioning had included also the unobserved covariate.

The differences-in-differences are slightly but not significantly longer in
the 700 blocks of type 2, with an increase in the replacement rate, than in
the 700 blocks of type 5, without an increase: \ Wilcoxon's rank sum test
yields a $P$-value of 0.13 with a Hodge-Lehmann estimated difference of 1.28
weeks and a 95\% confidence interval of $\left[ -0.43,\,3.06\right] $ weeks.

\vspace{-.5cm}

\section{Construction of Balanced Blocks of Maximal Size}

\label{secMatching}

The block design described in Table 4 was assembled in three steps, building on the method of cardinality matching \citep{zubizarreta2014matching}. \ Steps 1 and 2 use an extension of cardinality
matching to select units; then, Step 3 uses minimum distance matching to put the
selected units into blocks so that the blocks are as homogenous as
possible with respect to covariates. \ Steps 1 and 2 constitute a $P$-way template-balanced partitioning program, where each treatment group is partitioned into $P$ samples of maximal size balanced with respect to a template. \ As seen in Figures 1-2, Table 5 and in
other tables in the appendix, the observed covariates exhibit good balance,
indeed better balance than is expected from complete randomization of these
same individuals to these same groups. \ Of course, as always, balance for
observed covariate does not imply balance for covariates that were not
observed. \ We first briefly review cardinality
matching, then explain the three steps. 

In forming a single matched sample, cardinality matching picks the largest
control group that satisfies certain requirements for covariate balance. \ 
\cite{niknam2022using} provide an informal and practical exposition of
cardinality matching. \ Cardinality matching differs from minimum distance
matching, in which the size of the design is fixed in advance and a minimum
distance match of that size is constructed. \ Both are discrete optimization
problems, but in cardinality matching the sample size is maximized subject
to constraints on covariate balance, while in minimum distance matching a
total within-block covariate distance is minimized subject to constraints on
sample size and covariate balance; see \cite{rosenbaum2023optimization}.

In the first applications of cardinality matching, controls were picked to
resemble a treated group. \ More recently, controls have been picked to
resemble summary statistics describing the distribution of covariates in a
target population called a template; see \cite{cohn2022profile} and \cite%
{silber2014template}. \ This creates the possibility of building more
complex designs, essentially because the size of the template no longer
plays a role in cardinality matching.

In cardinality matching, a balance constraint is a linear inequality
constraint. \ Let $s$ be the size of the matched control group. \ Two
inequality constraints can express a constraint on an absolute value:
specify an $\epsilon >0$; then two inequalities, $A\leq s\epsilon $ and $%
-A\leq s\epsilon $, jointly require $\left\vert A\right\vert \leq s\epsilon $%
. \ For example, we might require the mean age among selected controls to
differ from the mean age in the template by at most one year, $\epsilon =1$,
and that requirement $A$ can be expressed as as a requirement on a linear
quantity, $\left\vert A\right\vert \leq s\epsilon $, where $%
A=sB-\sum_{i=1}^{I}c_{i}\,a_{i}$, $B$ is the mean age in the template, $%
a_{i} $ is the age of the $i$th of $I$ available controls, $c_{i}=1$ if
control $i$ is selected into the control group and $c_{i}=0$ otherwise, and $%
s=\sum_{i=1}^{I}c_{i}$. \ There are many such linear requirements for many
covariates, all with the same $c_{i}$ and $s$. \ Cardinality matching picks
the $c_{i}\in \left\{ 0,\,1\right\} $ to maximize $s=\sum_{i=1}^{I}c_{i}$
subject to these linear constraints.

Building the design in Table 4 requires an adjustment to cardinality
matching. \ There are $4\times 2=8$ groups: $\mathrm{BR}$, $\mathrm{Br}$, $%
\mathrm{bR}$, $\mathrm{br}$, $\overline{\mathrm{BR}}$, $\overline{\mathrm{Br}%
}$, $\overline{\mathrm{bR}}$ and $\overline{\mathrm{br}}$. \ Focus on one
group, say $\mathrm{BR}$. \ Group $\mathrm{BR}$ needs to be represented
three times in three separate block types of size $s$, one for blocks of
type 1, one for blocks of type 2, and one for blocks of type 3. \ The
situation is essentially the same for the other 7 groups. \ Focus on group $%
\mathrm{BR}$ which contains $I$ individuals before matching. \ Let $c_{ip}=1$
if individual $i$ from BR is selected for group $p$, $p=1$, \ldots , $P$,
where $P=3$. \ To ensure that a member of BR is selected for at most one
group, there is a linear constraint $c_{i1}+\cdots +c_{iP}\leq 1$ for $%
i=1,\ldots ,\,I$. \ For the covariate age, there are now $P=3$ constraints
instead of one, $\left\vert A_{p}\right\vert \leq s\epsilon $ with $%
A_{p}=sB-\sum_{i=1}^{I}c_{ip}\,a_{i}$ for $p=1,\ldots ,\,P$. \ Finally, the $%
P$ groups must each be of size $s$, so there are $P$ size constraints, $%
s=\sum_{i=1}^{I}c_{ip}$ for $p=1,\ldots ,\,P$. \ The problem in step 1 is to
pick $c_{ip}\in \left\{ 0,\,1\right\} $ to maximize $s$ subject to the
stated constraints. \ That produces $P=3$ groups from $\mathrm{BR}$ that
each resemble the template in terms of the distribution of covariates. \ 

The problem just described is solved $4\times 2=8$ times, for $\mathrm{BR}$, 
$\mathrm{Br}$, $\mathrm{bR}$, $\mathrm{br}$, $\overline{\mathrm{BR}}$, $%
\overline{\mathrm{Br}}$, $\overline{\mathrm{bR}}$ and $\overline{\mathrm{br}}
$. This yields eight values of $s$, and the minimum value, $\overline{s}$,
of $s$ is selected. \ This completes Step 1.

Step 2 builds $4\times 2=8$ balanced groups all of the same size, $\overline{%
s}$. \ In Step 2, for each of the $4\times 2=8$ groups, the problem in Step
1 is solved again, but now with the size $\overline{s}$ that was determined
in Step 1; that is, with the fixed constraint $\overline{s}%
=\sum_{i=1}^{I}c_{ip}$ for $p=1,\ldots ,\,P$. \ This can be done in several
ways that differ slightly and are described in the Appendix. \ For example,
with $\overline{s}$ fixed, one could optimize another quantity, such as the
sum of suitably scaled $\epsilon $'s over all of the covariates. \ See the
appendix for options and details.

At the end of Step 2, there are 24 balanced groups. \
Four balanced groups are assigned to block types 1 to 6 according to the
plan in Table 4. \ Step 3 is to assemble the $4\overline{s}$ individuals in
each block type into $4\overline{s}$ blocks of 4 individuals, one from each
of the 4 balanced groups in that block type. \ Within block type 2, $%
\overline{\mathrm{BR}}$ is optimally paired with $\mathrm{BR}$, to minimize
a rank-based Mahalanobis covariate distance \cite[\S 9.3]%
{rosenbaum2020design}. \ Importantly, because the eligibility variables ---
i.e., $w^{\prime }=$ LE and $w^{\prime \prime }=$ IU and their continuous
counterparts Prior Wage and Relative Employment in Table 5 --- are \textit{%
not} aliased for $\overline{\mathrm{BR}}$ and $\mathrm{BR}$, the
before-after pairing of BR to produce $\overline{\mathrm{BR}}-\mathrm{BR}$\
pairs \textit{can} pair for these covariates. \ Similarly, within block type
2, $\overline{\mathrm{bR}}$ is optimally paired with $\mathrm{bR}$. \ Then
the $\overline{\mathrm{BR}}-\mathrm{BR}$ pairs are paired with the $%
\overline{\mathrm{bR}}-\mathrm{bR}$ pairs.
\ In pairing pairs, the distance between two pairs is the sum of the four
covariate distances that cross the pairs (c.f.,  \cite{nattino2021triplet}). \ In pairing pairs, some
eligibility covariates \textit{cannot} be included, depending upon the block
type.

\section{Discussion: Recap and Extensions}
\vspace{-.5cm}
\subsection{Deterministic Treatment Assignment and Aliasing}

With 50 covariates, there are $50+\binom{50}{2}=1275$ main effects and two
factor interactions, and there are $2^{50}-1=1.1\times 10^{15}$ main effects
and interactions of all kinds. \ In modelling to predict outcomes
from covariates or matching to balance covariates, we should
consider low-order interactions, but we have no realistic prospect of
evaluating all interactions.

Here we have considered study designs in which treatment assignment is a
deterministic function of some covariates, yet certain contrasts among
treatment groups are immune to the presence of certain interactions among
covariates, including some covariates that determine treatment. \ If the
four aliased covariates are excluded from Table 5, then 17 covariates
remain. \ In that case, expression (\ref{eqSpecialAlias}) can represent the
sum of $2^{17}\times \left( 1+3+2\right) =786432$ main effects and
interactions that are not aliased with contrast $\mathbf{h}$ in Table 3,
even though the two factor interaction $w^{\prime \prime }\times w^{\prime
\prime \prime }$ is aliased, as is every interaction of $w^{\prime \prime
}\times w^{\prime \prime \prime }$ with the 17 covariates. \ Propositions %
\ref{propEstimateContrast} and \ref{eqIwithAciPropensity} give conditions
that a certain contrast among treatment groups is estimable despite
deterministic assignment to those groups on the basis of covariates.

A main effect or interaction in (\ref{eqSpecialAlias}) that is not aliased
may nonetheless be imbalanced or confounded, but Tables 5 and 6 and tables
in the on-line appendix check many such covariates and interactions for
balance after blocking, without finding any evidence of imbalance. \ Of
course, covariate interactions that are aliased are, by definition, severely
imbalanced. \ For covariates and their interactions that are not aliased, it
is straightforward and appropriate to check that they are balanced by a
group contrast like $\mathbf{h}$ in Table 3; however, this is not commonly
done.

\subsection{Other Aspects of Classical Experiment Design}

We have illustrated the construction of an observational study from
administrative data with reference to some principles of classical
experimental design, including principles of aliasing and incomplete blocks.
\ There has been active work on experimental design by statisticians for
about a century, and our discussion of analogous observational studies has
considered only a few aspects.

One traditional topic that we have not explored is the so-called
\textquotedblleft recovery of inter-block information.\textquotedblright\ \
We estimated the BR-versus-bR effect directly within-blocks using blocks of
type 2 in Table 4. \ However, if one calculated the difference-in-differences in blocks of types 1 and 4, then the shift in the two
distributions of those two groups of differences-in-differences would also
estimate the B-versus-b effect, because Br appears twice and cancels. \ That
inter-block estimate places different requirements on $\beta \left( \mathbf{x%
},\mathbf{w}\right) $ and compares individuals in different blocks of four
individuals, but there may be some useful information there nonetheless.

\subsection{Aspects Outside Classical Experimental Design}

Certain aspects of our blocked design appear to have no analogous form in
experimental design, essentially because the observational treatment
assignment is found by the investigator, not created by the investigator. \
An investigator would not create some of the structures that she may find. \
In particular, as noted in \S \ref{secMatching}, the before-versus-after
pairing of individuals in the same $w^{\prime }\times w^{\prime \prime }$
category can and did match for both binary and continuous aliased covariates
(LE and IU, and Prior Wage and Relative Employment) that lack common support
as $w^{\prime }\times w^{\prime \prime }$ varies. \ One might say that in a
difference-in-difference, the inter-temporal difference is more closely
matched than the inter-eligibility difference.

\subsection{Observational Studies with More Complex Aliasing}

\label{ssMoreComplexA}

Is complex aliasing rare in observational studies? \ Is \cite%
{lalive2006changes}'s study unusual in having aliasing that is more complex
than conventional difference-in-differences? \ We suspect that complex
aliasing is often unnoticed, but not unusual. \ We most readily recognize in
data those patterns that we have already seen in mathematical structures. 

Even \cite{lalive2006changes}'s study can be represented as having more
complex aliasing than we have discussed. \ In on-line Appendix Table 1, a
B-versus-b increase in benefit duration varies in size from adding as little
as 9 weeks of additional benefits to as much as 32 weeks depending upon two
other covariates; so, this variation in treatment is not aliased with $%
w^{\prime }$ and $w^{\prime \prime }$ in Table 3, but is nonetheless aliased
with certain other covariates (which are balanced in Tables 5
and 6). \ If benefit duration affects the length of subsequent unemployment,
then one might expect to see evidence of that by comparing an increase from
30 to 39 weeks with an increase from 20 to 52 weeks; yet, this comparison
takes place within cells defined by $w^{\prime }\times w^{\prime \prime
}\times w^{\prime \prime \prime }$, and so is free of some of the aliasing
that we have discussed.

\vspace{-.5cm}

\bibliographystyle{agsm}

\bibliography{alias}

\include{tablesFigures}

\include{appendix}

\end{document}

%% file: tablesFigures.tex
\singlespacing

\begin{table}[tbp] \centering%
\caption{A complete $2^3$ factorial and a 1/2 fraction in which main effects are not aliased with each other 
but are aliased with interactions.}\label{tabHalfFraction}$%
\begin{array}{rrrr}
i & A & B & C \\ 
1 & 1 & 1 & 1 \\ 
2 & 1 & 1 & -1 \\ 
3 & 1 & -1 & 1 \\ 
4 & 1 & -1 & -1 \\ 
5 & -1 & 1 & 1 \\ 
6 & -1 & 1 & -1 \\ 
7 & -1 & -1 & 1 \\ 
8 & -1 & -1 & -1%
\end{array}%
\Longrightarrow 
\begin{array}{rrrr}
i & A & B & C \\ 
1 & 1 & 1 & 1 \\ 
4 & 1 & -1 & -1 \\ 
6 & -1 & 1 & -1 \\ 
7 & -1 & -1 & 1%
\end{array}%
$
\end{table}%

\bigskip 

\begin{table}[tbp] \centering%
\caption{Design matrix  for difference-in-differences.  Because 
$w_{3}=(w_{1}+w_{1}\times w_{2})/2$, Treatment would be aliased if the
$w_{1} \times w_{2}$ interaction column were added to this design matrix.  Coding
``Treatment'' column as $(0,0,0,1)$ produces a
contrast that is proportional to $\mathbf{h}$.}\label{tabDinD}%
\begin{tabular}{rrrrrr}
$\mathbf{h}$ &  & $\mathbf{1}$ & Eligible, $w_{1}$ & 
Time, $w_{2}$ & Treatment, $w_{3}$ \\ \hline
1/2 & Ineligible-Before & 1 & -1 & -1 & 0 \\ 
-1/2 & Eligible-Before & 1 & 1 & -1 & 0 \\ 
-1/2 & Ineligible-After & 1 & -1 & 1 & -1 \\ 
1/2 & Eligible-After & 1 & 1 & 1 & 1%
\end{tabular}
\end{table}%

\begin{table}[tbp] \centering%
\caption{The design matrix for the observational study by Lalive et al. (2006).  Notice that this
design matrix exhibits less aliasing than the difference-in-difference design in Table 2.  For example, the longer 
benefits, B-vs-b, column is orthogonal to the interaction $w^{\prime} \times w^{\prime\prime\prime}$
for low earnings by time.}\label{tabDesignMatrix}$%
\begin{tabular}{rrccrrrrrr}
\hline
g & $\mathbf{h}$ &  &  & $\mathbf{1}$ & $\mathrm{B/b}$ & $\mathrm{R/r}$ & 
$\mathrm{LE}\text{, }w^{\prime }$ & $\mathrm{IU}\text{, }w^{\prime \prime }$ & 
$\mathrm{TIME}\text{, }w^{\prime \prime \prime }$ \\ \hline \\
1 & -1/4 & $\overline{\mathrm{BR}}$ &  & 1 & 0 & 0 & 1 & 1 & -1 \\ 
2 & -1/4 & $\overline{\mathrm{Br}}$ &  & 1 & 0 & 0 & -1 & 1 & -1 \\ 
3 & 1/4 & $\overline{\mathrm{bR}}$ &  & 1 & 0 & 0 & 1 & -1 & -1 \\ 
4 & 1/4 & $\overline{\mathrm{br}}$ &  & 1 & 0 & 0 & -1 & -1 & -1 \\ 
5 & 1/4 & $\mathrm{BR}$ &  & 1 & 1 & 1 & 1 & 1 & 1 \\ 
6 & 1/4 & $\mathrm{Br}$ &  & 1 & 1 & -1 & -1 & 1 & 1 \\ 
7 & -1/4 & $\mathrm{bR}$ &  & 1 & -1 & 1 & 1 & -1 & 1 \\ 
8 & -1/4 & $\mathrm{br}$ &  & 1 & -1 & -1 & -1 & -1 & 1 \\ \hline
\end{tabular}%
$%
\end{table}%

\begin{figure}
\centering
\includegraphics[scale = .8]{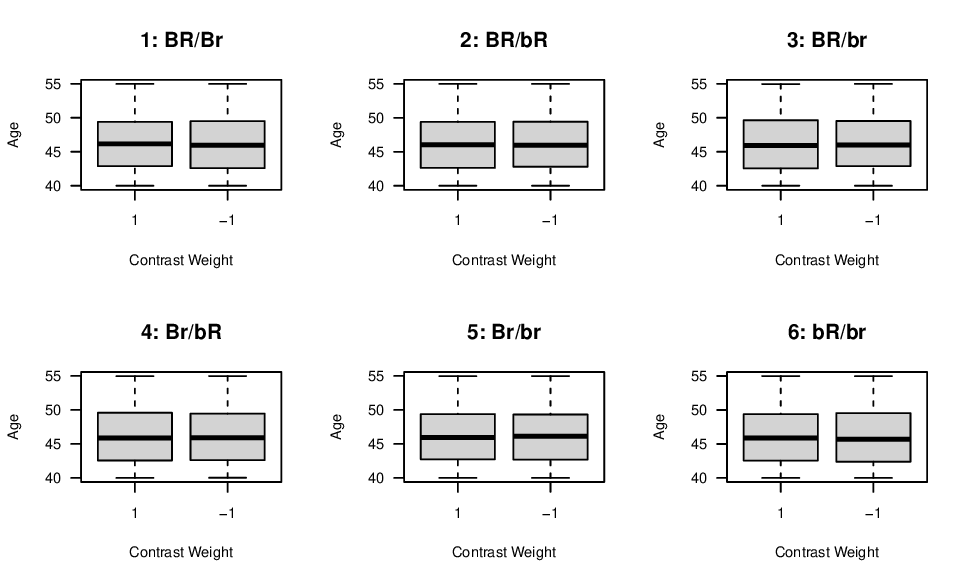}
\caption{Balance of the covariate ``age'' by contrast value, $+1$ or $-1$, in each of six type of blocks.}
\label{fig1}
\end{figure}

\begin{figure}
\centering
\includegraphics[scale = .8]{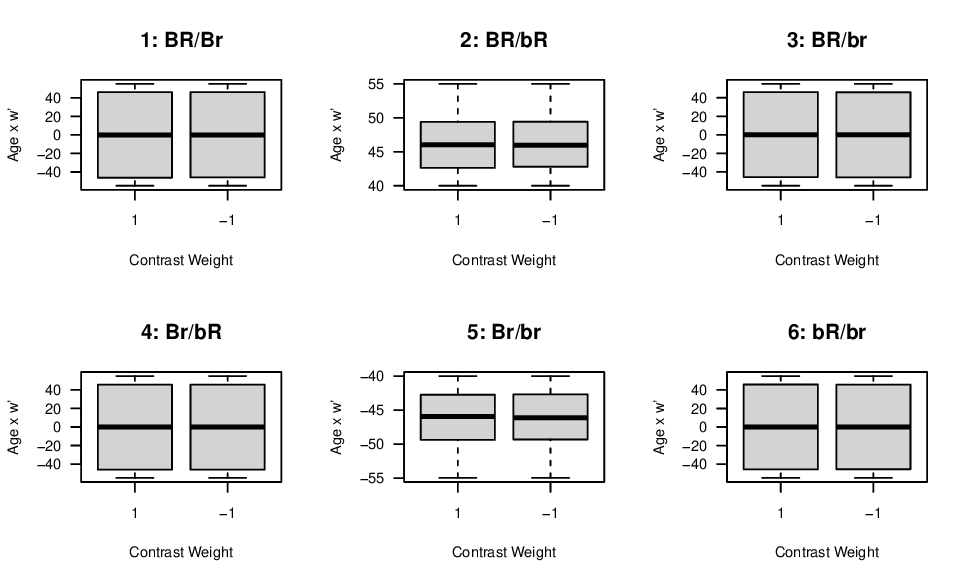}
\caption{Balance of the covariate ``age $\times$ LE'' or $x_{3} \times w^{\prime} $ by contrast value, $+1$ or $-1$, in each of six type of blocks.}
\label{fig2}
\end{figure}

\begin{table}[ht]
\centering
\caption{Six block types, 1 to 6, each
containing
two eligibility categories $\overline{\mathrm{Before}}$ and After two treatments were
applied in the After period.
Each block type occurs 700 times, so there are $4 \times
700 = 2800$ individuals in a
column, $6 \times
700 = 4200$ in a row, and $2800 \times 6 =16800$
in the study.}
\label{tabDesign}
\begin{tabular}{l|r|cccccc|r}
  \hline

&  & \multicolumn{6}{|c|}{ Eight Treatments in 6 Types of Blocks of Size 4} &  N
\\ \hline
& & 1 & 2 & 3 & 4 & 5 & 6 & \\
Time & $h_{g}$ & BR/Br & BR/bR & BR/br & Br/bR & Br/br & 
bR/br & \\ 
  \hline
After & 1 & $\mathrm{BR}$ &
$\mathrm{BR}$ & $\mathrm{BR}$ & $\mathrm{Br}$ & $\mathrm{Br}$ & $\mathrm{bR}$ & 4200 \\ 
After & $-1$ &
$\mathrm{Br}$ & $\mathrm{bR}$ & $\mathrm{br}$ & $\mathrm{bR}$ & $\mathrm{br}$ & $\mathrm{br}$ & 4200 \\

$\overline{\mathrm{Before}}$ & $-1$ & $\overline{\mathrm{BR}}$ & $\overline{\mathrm{BR}}$ & $\overline{\mathrm{BR}}$ &
$\overline{\mathrm{Br}}$ & $\overline{\mathrm{Br}}$ & $\overline{\mathrm{bR}}$ & 4200 \\

$\overline{\mathrm{Before}}$ & 1 & $\overline{\mathrm{Br}}$ & $\overline{\mathrm{bR}}$ &
$\overline{\mathrm{br}}$ & $\overline{\mathrm{bR}}$ & $\overline{\mathrm{br}}$ & $\overline{\mathrm{br}}$ & 4200 \\ 
\hline

N & & 2800 & 2800 & 2800 & 2800 & 2800 & 2800 & 16800 \\
   \hline
\end{tabular}
\end{table}%

\begin{table}[ht]
\centering
\caption{P-values checking covariate imbalance, $h_{g}=1$ versus $h_{g}=-1$,
compared with complete randomization, in each of 6 block types, plus
\textit{combining} (C) across 6 block types using the truncated product method.}
\label{tabCovBal}
\begin{tabular}{r|rrrrrr|r}
  \hline
Covariate & \multicolumn{6}{|c|}{ 6 Types of Blocks } &  \\ \hline
 & 1 & 2 & 3 & 4 & 5 & 6 & C \\ 
 & BR/Br & BR/bR & BR/br & Br/bR & Br/br & 
bR/br & \\
  \hline
Temporary Layoff & 0.78 & 0.60 & 0.94 & 0.90 & 0.68 & 0.65 & 1.00 \\ 
  Female & 0.94 & 1.00 & 0.82 & 1.00 & 0.97 & 0.88 & 1.00 \\ 
  Age & 0.58 & 0.99 & 0.67 & 0.94 & 0.77 & 0.85 & 1.00 \\ 
  Secondary Education & 1.00 & 0.24 & 0.86 & 0.74 & 0.36 & 0.49 & 1.00 \\ 
  Tertiary Education & 1.00 & 1.00 & 1.00 & 1.00 & 1.00 & 1.00 & 1.00 \\ 
  Apprenticeship & 0.90 & 0.67 & 0.70 & 0.58 & 0.83 & 0.87 & 1.00 \\ 
  Married & 0.81 & 0.97 & 0.94 & 0.90 & 0.97 & 0.94 & 1.00 \\ 
  Single & 0.46 & 0.95 & 0.61 & 0.87 & 0.65 & 0.91 & 1.00 \\ 
  Divorced & 1.00 & 0.92 & 0.67 & 0.79 & 0.75 & 0.79 & 1.00 \\ 
  Female x Married & 0.82 & 0.33 & 1.00 & 0.85 & 0.75 & 0.76 & 1.00 \\ 
  Female x Single & 0.93 & 0.42 & 1.00 & 1.00 & 0.93 & 0.92 & 1.00 \\ 
  Female x Divorced & 0.69 & 0.64 & 0.28 & 1.00 & 0.69 & 0.90 & 1.00 \\ 
  Blue Collar Job & 0.63 & 0.79 & 0.97 & 0.54 & 0.90 & 0.57 & 1.00 \\ 
  Seasonal Job & 0.70 & 0.97 & 0.97 & 0.94 & 0.49 & 0.64 & 1.00 \\ 
  Manufacturing Job & 0.72 & 0.57 & 0.96 & 0.96 & 0.58 & 0.53 & 1.00 \\ 
  Prior Wage & 0.94 & 0.52 & 0.66 & 0.88 & 0.91 & 0.79 & 1.00 \\ 
  (LE) Prior Wage $\leq$ 12610 & 1.00 & 1.00 & 1.00 & 1.00 & 1.00 & 1.00 & 1.00 \\ 
  Relative Employment (RE) & 0.90 & 0.98 & 0.85 & 0.91 & 0.93 & 0.95 & 1.00 \\ 
  (IU) RE $\geq$ 40\% & 1.00 & 1.00 & 1.00 & 1.00 & 1.00 & 1.00 & 1.00 \\ 
  Worked 3 of 5 Years & 0.97 & 1.00 & 0.97 & 0.81 & 0.81 & 0.97 & 1.00 \\ 
  Age $\geq$ 50 & 0.71 & 0.93 & 0.96 & 0.71 & 0.89 & 0.68 & 1.00 \\
   \hline
\end{tabular}
\end{table}

\begin{figure}
\centering
\includegraphics[scale = .82]{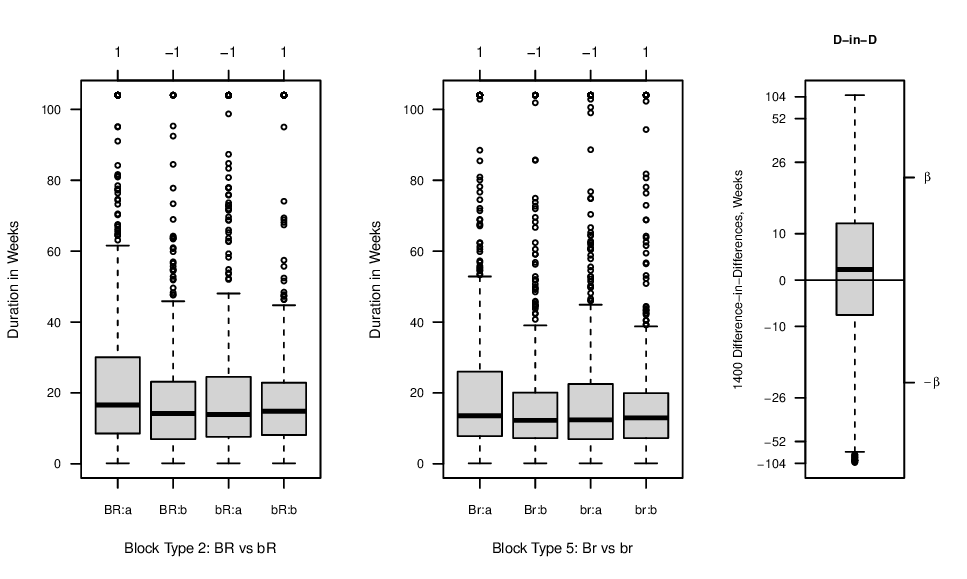}
\caption{Duration of unemployment, with and without an increase in benefits duration (B/b) in block type 2 at high R and block type 5 at low R, after (a) and before (b). 
Contrast weights, $h_{g}=1$ or $h_{g}=-1$, appear at the top.  The right boxplot
shows $1400 = 700 + 700$ difference-in-difference estimates, pooling blocks
of types 2 and 5, with symmetric transformation of the tails
to better visualize the center of the distribution.}
\label{fig3}
\end{figure}

\begin{table}[h]
\centering
\caption{P-values checking covariate imbalance, $h_{g}=1$ versus $h_{g}=-1$,
for covariates interacted with LE, $x_{k} \times w^{\prime}$, compared with complete randomization, in each of 6 block types, plus
combining across 6 block types (C).}
\label{tabCovBal2}
\begin{tabular}{r|rrrrrr|r}
  \hline
Covariate & \multicolumn{6}{|c|}{ 6 Types of Blocks } &  \\ \hline
 & 1 & 2 & 3 & 4 & 5 & 6 & C \\ 
 & BR/Br & BR/bR & BR/br & Br/bR & Br/br & 
bR/br & \\
  \hline
Temporary Layoff & 0.62 & 0.60 & 0.33 & 0.38 & 0.68 & 0.73 & 1.00 \\ 
  Female & 0.94 & 1.00 & 0.76 & 1.00 & 0.97 & 0.88 & 1.00 \\ 
  Age & 0.87 & 0.99 & 0.81 & 0.94 & 0.77 & 0.97 & 1.00 \\ 
  Secondary Education & 0.94 & 0.24 & 0.41 & 0.88 & 0.36 & 0.55 & 1.00 \\ 
  Tertiary Education & 1.00 & 1.00 & 1.00 & 1.00 & 1.00 & 1.00 & 1.00 \\ 
  Apprenticeship & 0.62 & 0.67 & 0.79 & 0.57 & 0.83 & 0.82 & 1.00 \\ 
  Married & 0.82 & 0.97 & 0.94 & 0.91 & 0.97 & 0.94 & 1.00 \\ 
  Single & 0.91 & 0.95 & 0.79 & 0.85 & 0.65 & 1.00 & 1.00 \\ 
  Divorced & 0.71 & 0.92 & 0.88 & 0.85 & 0.75 & 0.85 & 1.00 \\ 
  Female x Married & 0.88 & 0.33 & 0.94 & 0.52 & 0.75 & 1.00 & 1.00 \\ 
  Female x Single & 0.97 & 0.42 & 0.71 & 1.00 & 0.93 & 0.97 & 1.00 \\ 
  Female x Divorced & 0.65 & 0.64 & 0.82 & 0.82 & 0.69 & 0.82 & 1.00 \\ 
  Blue Collar Job & 0.73 & 0.79 & 0.97 & 0.60 & 0.90 & 0.68 & 1.00 \\ 
  Seasonal Job & 0.71 & 0.97 & 0.97 & 0.94 & 0.49 & 0.65 & 1.00 \\ 
  Manufacturing Job & 0.79 & 0.57 & 0.62 & 0.97 & 0.58 & 0.65 & 1.00 \\ 
  Prior Wage & 0.81 & 0.52 & 0.71 & 0.81 & 0.91 & 0.95 & 1.00 \\ 
  (LE) Prior Wage $\leq$ 12610 & 1.00 & 1.00 & 1.00 & 1.00 & 1.00 & 1.00 & 1.00 \\ 
  Relative Employment (RE) & 0.64 & 0.98 & 0.88 & 0.88 & 0.93 & 0.99 & 1.00 \\ 
  (IU) RE $\geq$ 40\% & 1.00 & 1.00 & 1.00 & 1.00 & 1.00 & 1.00 & 1.00 \\ 
  Worked 3 of 5 Years & 0.97 & 1.00 & 0.97 & 0.82 & 0.81 & 0.97 & 1.00 \\ 
  Age $\geq$ 50 & 0.82 & 0.93 & 0.97 & 0.50 & 0.89 & 0.97 & 1.00 \\ 
   \hline
\end{tabular}
\end{table}

%% file: appendix.tex
\singlespacing

\section*{Appendix 1: Benefit Details}

The benefit changes were fairly subtle, as described in Table \ref{tabBenefits}.

Whether benefit duration increased from before to after depended upon pnonGE40, which is called ``Relative Employment $\geq$ 40'' in the balance tables.  Depending upon age, pnonGE40 indicates whether you has worked for at least 60\% of the past 10 or 15 years.  If yes, you received an increase in benefit duration; otherwise, no increase.  The size of the increase, if any, was larger if you were at least 50 years old.

Whether the replacement rate was increased depended upon the wages in Austrian Shillings in the prior job, wagePJGE12610, which is called ``Prior Wage $\leq$ 12610''.  If your wages were low in your prior job, you received a larger amount of money.

Also important were the benefits available in the before period.  After all, a change in benefits is more likely to influence behavior if the levels before and after are very different.

Before the benefit changes, the benefit duration depended upon whether you had worked for at least 3 of the last 5 years (e3-5).  

The matching tried to balance all of these variables.

\section*{Appendix 2: Further Balance Checks}

Tables \ref{tabCovBal3} -- \ref{tabCovBal6} expand the checks on covariate balance to include interactions with $w^{\prime}$ or $w^{\prime\prime}$.  Tables \ref{tabCovBal3} -- \ref{tabCovBal4} check interactions that \textit{could} be balanced given the alias structure, and these tables show that the matching did, indeed, balance them.  Tables \ref{tabCovBal5} -- \ref{tabCovBal6} check interactions that \textit{cannot} be balanced given the alias structure, and these tables show that the matching did not balance them. When certain covariate combinations perfectly control treatment assignment, those covariate combinations must exibit imbalances.

\section*{Appendix 3: Three Step Matching Method}

The mean of covariate $k$ in the template is $B_{k}$. \ A tolerance, $%
\epsilon _{k}>0$ is specified for covariate $k$. \ For this purpose, it may
be best to standardize every covariate to have mean $B_{k}=0$ and standard
deviation $1$ in the template, so $\epsilon _{k}$ is the tolerable
difference in means in units of the standard deviation; however, this is for
ease of interpretation and is not essential. \ For instance, with this
standardization, one might set $\epsilon _{k}=0.05$ for all $k$, as this
would be a small difference in means which is 5\% of the standard deviation
in the template.

Step 1 solves the following integer program for each treatment group, where
there were $T=8=4\times 2$ treatment groups in the main paper. \ In the
paper, $P=3$. \ Other quantities, like $I$, vary with the treatment group. \
The value of covariate $k$ for person $i$ in this treatment group is $x_{ik}$

\pagebreak

The integer program is:

\[
\max s 
\]%
subject to%
\begin{eqnarray*}
s &=&\sum_{i=1}^{I}c_{ip}\text{, }p=1,\ldots ,\,P \\
1 &\geq &\sum_{p=1}^{P}c_{ip}\text{, }i=1,\ldots ,\,I\text{,} \\
\epsilon _{k}\,s &\geq &\left\vert
sB_{k}-\sum_{i=1}^{I}c_{ip}\,x_{ik}\right\vert \text{, }k=1,\ldots ,\,K\text{%
, }p=1,\ldots ,\,P\text{,} \\
c_{ip} &\in &\left\{ 0,1\right\} \text{, }i=1,\ldots ,\,I\text{, }p=1,\ldots
,\,P\text{.}
\end{eqnarray*}

This results in $T$ maximized values of $s$, one for each treatment group. \
The final group size is the minimum of these $T$ values of $s$, say $%
\overline{s}$.

In Step 2, $T$ groups of the same size $\overline{s}$ are selected. \ One
approach takes $\epsilon _{k}$ as a free variable and solves $T$ times the
following mixed integer program:%
\[
\min \sum_{k=1}^{K}\epsilon _{k} 
\]%
subject to%
\begin{eqnarray*}
\overline{s} &=&\sum_{i=1}^{I}c_{ip}\text{, }p=1,\ldots ,\,P \\
1 &\geq &\sum_{p=1}^{P}c_{ip}\text{, }i=1,\ldots ,\,I\text{,} \\
\epsilon _{k}\,\overline{s} &\geq &\left\vert \overline{s}%
B_{k}-\sum_{i=1}^{I}c_{ip}\,x_{ik}\right\vert \text{, }k=1,\ldots ,\,K\text{%
, }p=1,\ldots ,\,P\text{,} \\
c_{ip} &\in &\left\{ 0,1\right\} \text{, }i=1,\ldots ,\,I\text{, }p=1,\ldots
,\,P\text{,} \\
\epsilon _{k} &\geq &0\text{ for }k=1,\ldots ,\,K\text{.}
\end{eqnarray*}

Conceptually, this program is attractive for two reasons. \ First, if $s>%
\overline{s}$ in step 1, then better balance may be possible: it may be
possible to reduce many $\epsilon _{k}$ if we want only $P\overline{s}$
individuals rather than $Ps>P\overline{s}$ individuals. \ Second, in
principle, it is conceivable --- if unlikely --- that reducing the size from 
$s$ to $\overline{s}$ makes it harder to balance some covariate $x_{ik}$,
and this mixed integer program minimizes the extent of any problems of this
sort.

An alternative approach is to omit the objective function in this program
and simply search for a feasible solution, one that satisfies the
constraints. \ In this case, the prespecified tolerances $\epsilon _{k}$
from Step 1 are used. \ After all, the original goal in Step 1 was simply to
satisfy the specified balance constraints. \ We know from Step 1 that it is
impossible to satisfy the specified balance constraints for all $T$ groups
if $s>\overline{s}$; so, a feasible solution, if it exists, is the largest
design with groups of equal size that does satisfy the specified balance
constraints for all $T$ groups. \ This was the approach taken in the paper.

Step 3 is a sequence of entirely standard optimal assignment problems.

Optimization was performed using \texttt{gurobi} in \texttt{R}.

\begin{table}[ht]
\centering\caption
{Benefit changes from before (B) to after (A) August 1989, expressed as the number of months of
benefits, with a plus (+) signifying an increase in the monetary benefit amount or ``replacement rate''.  Treatment PBD means an
increase in the benefit duration, while RR means an increase in the monetary amount.  Covariates
$\mathcal{T}$ determine the treatment change that is applied in the period after August 1989.
Covariates $\mathcal{B}$ determine the benefits in place before August 1989, and hence
define the baseline for the change from before to after.}
\label{tabBenefits}
\begin{tabular}{rrr|rrrr}
\hline& &  & \multicolumn{4}{|c}{Covariates $\mathcal{T}$}  \\ \hline
\multicolumn{3}{c|}{Worked in 10 or 15 years (pnonGE40) }       &  $\geq
60\% $ &  $\geq60\% $ &  $< 60\% $ &  $< 60\% $ \\
\multicolumn{3}{c|}{Wage in Prior Job (wagePJGE12610)}       & $\leq
12610 $ & $ > 12610 $ & $\leq12610 $  & $ > 12610 $ \\ \hline\multicolumn
{2}{c|}{Covariates $\mathcal{B}$} &  & \multicolumn{4}{|c}{Treatment Groups}
\\ \hline
Age  & Worked Years (e3-5) & After & PBD+RR & PBD & RR & Control \\
\hline$< 50$ &  $<$ 3 of 5  & B & 20 & 20 & 20 & 20 \\
$< 50$ &  $<$ 3 of 5  & A & +39 & 39 & +20 & 20 \\ \hline$< 50$ &  $\geq
$ 3 of 5  & B & 30 & 30 & 30 & 30 \\
$< 50$ &  $\geq$ 3 of 5  & A & +39 & 39 & +30 & 30 \\ \hline$\geq
50$ &  $<$ 3 of 5  & B & 20 & 20 & 20 & 20 \\
$\geq50$ &  $<$ 3 of 5  & A & +52 & 52 & +20 & 20 \\ \hline$\geq
50$ &  $\geq$ 3 of 5  & B & 30 & 30 & 30 & 30 \\
$\geq50$ &  $\geq$ 3 of 5  & A & +52 & 52 & +30 & 30 \\ \hline
& & & \multicolumn{4}{|c}{Treatment Groups}  \\
& & B & $\overline{\mathrm{BR}}$ & $\overline{\mathrm{Br}}$ & $\overline
{\mathrm{bR}}$ & $\overline{\mathrm{br}}$  \\
& & A & BR & Br & bR & br  \\
\hline\end{tabular}
\end{table}

\begin{table}[h]
\centering
\caption{P-values checking covariate imbalance, $h_{g}=1$ versus $h_{g}=-1$,
for covariates interacted with IU, $x_{k} \times w^{\prime\prime}$, 
compared with complete randomization, in each of 6 block types, plus
combining across 6 block types (C).}
\label{tabCovBal3}
\begin{tabular}{r|rrrrrr|r}
  \hline
Covariate & \multicolumn{6}{|c|}{ 6 Types of Blocks } &  \\ \hline
 & 1 & 2 & 3 & 4 & 5 & 6 & C \\ 
 & BR/Br & BR/bR & BR/br & Br/bR & Br/br & 
bR/br & \\
  \hline
Temporary Layoff & 0.78 & 0.57 & 0.33 & 0.38 & 0.50 & 0.65 & 1.00 \\ 
  Female & 0.94 & 0.94 & 0.76 & 1.00 & 0.91 & 0.88 & 1.00 \\ 
  Age & 0.58 & 0.93 & 0.81 & 0.94 & 0.98 & 0.85 & 1.00 \\ 
  Secondary Education & 1.00 & 0.62 & 0.41 & 0.88 & 0.73 & 0.49 & 1.00 \\ 
  Tertiary Education & 1.00 & 1.00 & 1.00 & 1.00 & 1.00 & 1.00 & 1.00 \\ 
  Apprenticeship & 0.90 & 0.71 & 0.79 & 0.57 & 0.57 & 0.87 & 1.00 \\ 
  Married & 0.81 & 0.97 & 0.94 & 0.91 & 0.91 & 0.94 & 1.00 \\ 
  Single & 0.46 & 0.97 & 0.79 & 0.85 & 0.82 & 0.91 & 1.00 \\ 
  Divorced & 1.00 & 0.94 & 0.88 & 0.85 & 0.82 & 0.79 & 1.00 \\ 
  Female x Married & 0.82 & 0.97 & 0.94 & 0.52 & 0.76 & 0.76 & 1.00 \\ 
  Female x Single & 0.93 & 0.73 & 0.71 & 1.00 & 0.97 & 0.92 & 1.00 \\ 
  Female x Divorced & 0.69 & 0.62 & 0.82 & 0.82 & 1.00 & 0.90 & 1.00 \\ 
  Blue Collar Job & 0.63 & 0.55 & 0.97 & 0.60 & 0.97 & 0.57 & 1.00 \\ 
  Seasonal Job & 0.70 & 0.97 & 0.97 & 0.94 & 0.71 & 0.64 & 1.00 \\ 
  Manufacturing Job & 0.72 & 0.73 & 0.62 & 0.97 & 0.68 & 0.53 & 1.00 \\ 
  Prior Wage & 0.94 & 0.93 & 0.71 & 0.81 & 0.72 & 0.79 & 1.00 \\ 
  (LE) Prior Wage $\leq$ 12610 & 1.00 & 1.00 & 1.00 & 1.00 & 1.00 & 1.00 & 1.00 \\ 
  Relative Employment (RE) & 0.90 & 0.88 & 0.88 & 0.88 & 0.84 & 0.95 & 1.00 \\ 
  (IU) RE $\geq$ 40\% & 1.00 & 1.00 & 1.00 & 1.00 & 1.00 & 1.00 & 1.00 \\ 
  Worked 3 of 5 Years & 0.97 & 1.00 & 0.97 & 0.82 & 0.82 & 0.97 & 1.00 \\ 
  Age $\geq$ 50 & 0.71 & 0.94 & 0.97 & 0.50 & 0.62 & 0.68 & 1.00 \\ 
   \hline
\end{tabular}
\end{table}

\begin{table}[h]
\centering
\caption{P-values checking covariate imbalance, $h_{g}=1$ versus $h_{g}=-1$,
for covariates interacted with \textit{both} LE and IU, $x_{k} \times w^{\prime} \times w^{\prime\prime}$, 
compared with complete randomization, in each of 6 block types, plus
combining across 6 block types (C).}
\label{tabCovBal4}
\begin{tabular}{r|rrrrrr|r}
  \hline
Covariate & \multicolumn{6}{|c|}{ 6 Types of Blocks } &  \\ \hline
 & 1 & 2 & 3 & 4 & 5 & 6 & C \\ 
 & BR/Br & BR/bR & BR/br & Br/bR & Br/br & 
bR/br & \\ 
  \hline
Temporary Layoff & 0.62 & 0.57 & 0.94 & 0.90 & 0.50 & 0.73 & 1.00 \\ 
  Female & 0.94 & 0.94 & 0.82 & 1.00 & 0.91 & 0.88 & 1.00 \\ 
  Age & 0.87 & 0.93 & 0.67 & 0.94 & 0.98 & 0.97 & 1.00 \\ 
  Secondary Education & 0.94 & 0.62 & 0.86 & 0.74 & 0.73 & 0.55 & 1.00 \\ 
  Tertiary Education & 1.00 & 1.00 & 1.00 & 1.00 & 1.00 & 1.00 & 1.00 \\ 
  Apprenticeship & 0.62 & 0.71 & 0.70 & 0.58 & 0.57 & 0.82 & 1.00 \\ 
  Married & 0.82 & 0.97 & 0.94 & 0.90 & 0.91 & 0.94 & 1.00 \\ 
  Single & 0.91 & 0.97 & 0.61 & 0.87 & 0.82 & 1.00 & 1.00 \\ 
  Divorced & 0.71 & 0.94 & 0.67 & 0.79 & 0.82 & 0.85 & 1.00 \\ 
  Female x Married & 0.88 & 0.97 & 1.00 & 0.85 & 0.76 & 1.00 & 1.00 \\ 
  Female x Single & 0.97 & 0.73 & 1.00 & 1.00 & 0.97 & 0.97 & 1.00 \\ 
  Female x Divorced & 0.65 & 0.62 & 0.28 & 1.00 & 1.00 & 0.82 & 1.00 \\ 
  Blue Collar Job & 0.73 & 0.55 & 0.97 & 0.54 & 0.97 & 0.68 & 1.00 \\ 
  Seasonal Job & 0.71 & 0.97 & 0.97 & 0.94 & 0.71 & 0.65 & 1.00 \\ 
  Manufacturing Job & 0.79 & 0.73 & 0.96 & 0.96 & 0.68 & 0.65 & 1.00 \\ 
  Prior Wage & 0.81 & 0.93 & 0.66 & 0.88 & 0.72 & 0.95 & 1.00 \\ 
  (LE) Prior Wage $\leq$ 12610 & 1.00 & 1.00 & 1.00 & 1.00 & 1.00 & 1.00 & 1.00 \\ 
  Relative Employment (RE) & 0.64 & 0.88 & 0.85 & 0.91 & 0.84 & 0.99 & 1.00 \\ 
  (IU) RE $\geq$ 40\% & 1.00 & 1.00 & 1.00 & 1.00 & 1.00 & 1.00 & 1.00 \\ 
  Worked 3 of 5 Years & 0.97 & 1.00 & 0.97 & 0.81 & 0.82 & 0.97 & 1.00 \\ 
  Age $\geq$ 50 & 0.82 & 0.94 & 0.96 & 0.71 & 0.62 & 0.97 & 1.00 \\ 
   \hline
\end{tabular}
\end{table}

\begin{table}[h]
\centering
\caption{Unlike previous balance tables, this table shows that certain interactions \textit{are} aliased, as described in the main manuscript.  P-values checking covariate imbalance, $h_{g}=1$ versus $h_{g}=-1$,
for covariates interacted with \textit{both} LE and Time, $x_{k} \times w^{\prime} \times w^{\prime\prime\prime}$, 
compared with complete randomization, in each of 6 block types, plus
combining across 6 block types (C).  Recall that $w^{\prime}$ is LE, and
$w^{\prime\prime}$ is IU.  There is covariate balance for block types 2 and 5, except for the RE/IU variables, but there is severe aliasing in other block types. In
block types 2 and 5, $w^{\prime}$ or LE is constant.}
\label{tabCovBal5}
\begin{tabular}{r|rrrrrr|r}
   \hline
Covariate & \multicolumn{6}{|c|}{ 6 Types of Blocks } &  \\ \hline
 & 1 & 2 & 3 & 4 & 5 & 6 & C \\ 
 & BR/Br & BR/bR & BR/br & Br/bR & Br/br & 
bR/br & \\
  \hline
Temporary Layoff & 0.00 & 0.62 & 0.00 & 0.00 & 0.71 & 0.00 & 0.00 \\ 
  Female & 0.00 & 1.00 & 0.00 & 0.00 & 0.97 & 0.00 & 0.00 \\ 
  Age & 0.00 & 0.60 & 0.00 & 0.00 & 0.64 & 0.00 & 0.00 \\ 
  Secondary Education & 0.00 & 0.57 & 0.00 & 0.00 & 0.57 & 0.00 & 0.00 \\ 
  Tertiary Education & 0.00 & 1.00 & 0.00 & 0.00 & 1.00 & 0.00 & 0.00 \\ 
  Apprenticeship & 0.00 & 0.55 & 0.00 & 0.00 & 0.73 & 0.00 &  0.00 \\ 
  Married & 0.00 & 0.97 & 0.00 & 0.00 & 0.43 & 0.00 & 0.00 \\ 
  Single & 0.00 & 0.52 & 0.00 & 0.00 & 0.55 & 0.00 & 0.00 \\ 
  Divorced & 0.00 & 0.33 & 0.00 & 0.00 & 0.55 & 0.00 & 0.00 \\ 
  Female x Married & 0.00 & 0.91 & 0.00 & 0.00 & 0.76 & 0.00 & 0.00 \\ 
  Female x Single & 0.00 & 0.47 & 0.00 & 0.00 & 0.62 & 0.00 & 0.00 \\ 
  Female x Divorced & 0.00 & 0.57 & 0.00 & 0.00 & 0.71 & 0.00 & 0.00 \\ 
  Blue Collar Job & 0.00 & 0.76 & 0.00 & 0.00 & 0.97 & 0.00 & 0.00 \\ 
  Seasonal Job & 0.00 & 0.97 & 0.00 & 0.00 & 0.88 & 0.00 & 0.00 \\ 
  Manufacturing Job & 0.00 & 0.52 & 0.00 & 0.00 & 0.62 & 0.00 & 0.00 \\ 
  Prior Wage & 0.00 & 0.95 & 0.00 & 0.00 & 0.45 & 0.00 & 0.00 \\ 
  (LE) Prior Wage $\leq$ 12610 & 1.00 & 1.00 & 1.00 & 1.00 & 1.00 & 1.00 & 1.00 \\ 
  Relative Employment (RE) & 0.00 & 0.00 & 0.00 & 0.00 & 0.00 & 0.00 & 0.00 \\ 
  (IU) RE $\geq$ 40\% & 0.00 & 0.00 & 1.00 & 1.00 & 0.00 & 0.00 & 0.00 \\ 
  Worked 3 of 5 Years & 0.00 & 0.11 & 0.00 & 0.00 & 0.20 & 0.00 & 0.00 \\ 
  Age $\geq$ 50 & 0.00 & 0.94 & 0.00 & 0.00 & 0.62 & 0.00 &  0.00\\
   \hline
\end{tabular}
\end{table}

\begin{table}[h]
\centering
\caption{Unlike most previous balance tables, this table shows that certain interactions \textit{are} aliased, as described in the main manuscript.  P-values checking covariate imbalance, $h_{g}=1$ versus $h_{g}=-1$,
for covariates interacted with \textit{both} IU and Time, $x_{k} \times w^{\prime\prime} \times w^{\prime\prime\prime}$, 
compared with complete randomization, in each of 6 block types, plus
combining across 6 block types (C).  Recall that $w^{\prime}$ is LE, and
$w^{\prime\prime}$ is IU.  There is covariate balance for block types 1 and 6, except for the Prior Wage/LE variables, but there is severe aliasing in other block types. In
block types 1 and 6, $w^{\prime\prime}$ or IU is constant.}
\label{tabCovBal6}
\begin{tabular}{r|rrrrrr|r}
   \hline
Covariate & \multicolumn{6}{|c|}{ 6 Types of Blocks } &  \\ \hline
 & 1 & 2 & 3 & 4 & 5 & 6 & C \\ 
 & BR/Br & BR/bR & BR/br & Br/bR & Br/br & 
bR/br & \\
  \hline
Temporary Layoff & 0.57 & 0.00 & 0.00 & 0.00 & 0.00 & 0.79 & 0.00  \\ 
  Female & 0.45 & 0.00 & 0.00 & 0.00 & 0.00 & 0.36 & 0.00 \\ 
  Age & 0.78 & 0.00 & 0.00 & 0.00 & 0.00 & 0.96 & 0.00  \\ 
  Secondary Education & 0.33 & 0.00 & 0.00 & 0.00 & 0.00 & 1.00 & 0.00  \\ 
  Tertiary Education & 1.00 & 0.00 & 0.00 & 0.00 & 0.00 & 1.00 & 0.00  \\ 
  Apprenticeship & 0.57 & 0.00 & 0.00 & 0.00 & 0.00 & 0.50 & 0.00  \\ 
  Married & 0.41 & 0.00 & 0.00 & 0.00 & 0.00 & 0.94 & 0.00 \\ 
  Single & 0.97 & 0.00 & 0.00 & 0.00 & 0.00 & 0.94 & 0.00  \\ 
  Divorced & 0.88 & 0.00 & 0.00 & 0.00 & 0.00 & 0.43 & 0.00  \\ 
  Female x Married & 0.45 & 0.00 & 0.00 & 0.00 & 0.00 & 0.45 & 0.00 \\ 
  Female x Single & 0.31 & 0.00 & 0.00 & 0.00 & 0.00 & 0.91 & 0.00  \\ 
  Female x Divorced & 0.88 & 0.00 & 0.00 & 0.00 & 0.00 & 0.65 & 0.00  \\ 
  Blue Collar Job & 0.73 & 0.00 & 0.00 & 0.00 & 0.00 & 0.62 & 0.00  \\ 
  Seasonal Job & 0.60 & 0.00 & 0.00 & 0.00 & 0.00 & 0.55 & 0.00 \\ 
  Manufacturing Job & 0.57 & 0.00 & 0.00 & 0.00 & 0.00 & 0.65 & 0.00  \\ 
  Prior Wage & 0.00 & 0.00 & 0.00 & 0.00 & 0.00 & 0.00 & 0.00  \\ 
  (LE) Prior Wage $\leq$ 12610 & 0.00 & 0.00 & 1.00 & 1.00 & 0.00 & 0.00 & 0.00 \\ 
  Relative Employment (RE) & 0.51 & 0.00 & 0.00 & 0.00 & 0.00 & 0.88 & 0.00  \\ 
  (IU) RE $\geq$ 40\% & 1.00 & 1.00 & 1.00 & 1.00 & 1.00 & 1.00 & 1.00 \\ 
  Worked 3 of 5 Years & 0.97 & 0.00 & 0.00 & 0.00 & 0.00 & 0.97 & 0.00 \\ 
  Age $\geq$ 50 & 0.82 & 0.00 & 0.00 & 0.00 & 0.00 & 0.97 & 0.00 \\
   \hline
\end{tabular}
\end{table}

\section*{Appendix 4: Construction of the Template}

As described above, each treatment group is partitioned into $P = 3$ disjoint samples of maximal and equal size that are balanced with respect to a summary of a reference distribution or template.
The template was computed as the simple means of the covariates across all treatment groups. 
The template is displayed in Table \ref{tabTemplate}.

\begin{table}[ht]
\centering
\caption{Description of the template, computed as the simple means of the covariates across all treatment groups.}
\label{tabTemplate}
\begin{tabular}{rr}
  \hline
Covariate & Mean \\ 
  \hline
Temporary Layoff & 0.32 \\ 
Female & 0.55 \\ 
Age & 46.30 \\ 
Secondary Education & 0.05 \\ 
Tertiary Education & 0.04 \\ 
Apprenticeship & 0.28 \\ 
Married & 0.68 \\ 
Single & 0.13 \\ 
Divorced & 0.15 \\ 
Female x Married & 0.38 \\ 
Female x Single & 0.05 \\ 
Female x Divorced & 0.09 \\ 
Blue Collar Job & 0.75 \\ 
Seasonal Job & 0.40 \\ 
Manufacturing Job & 0.16 \\ 
Age $>$= 50 x Worked 3 of 5 Years = 00 & 0.29 \\ 
Age $>$= 50 x Worked 3 of 5 Years = 01 & 0.50 \\ 
Age $>$= 50 x Worked 3 of 5 Years = 10 & 0.07 \\ 
Age $>$= 50 x Worked 3 of 5 Years = 11 & 0.15 \\ 
Relative Employment & 0.39 \\ 
Prior Wage & 14540.20 \\ 
   \hline
\end{tabular}
\end{table}

